\documentclass[journal]{IEEEtran}
\usepackage{amsmath,amsfonts}
\usepackage{amsthm}
\usepackage{algorithmic}
\usepackage{algorithm}
\usepackage{array}
\usepackage{tabularx}
\usepackage[caption=false,font=normalsize,labelfont=sf,textfont=sf]{subfig}
\usepackage{textcomp}
\usepackage{stfloats}
\usepackage{url}
\usepackage{verbatim}
\usepackage{graphicx}
\usepackage{cite}
\usepackage[hidelinks]{hyperref}
\usepackage{orcidlink}
\usepackage{tikz}
\usepackage{multirow}
\usepackage{placeins}
\usepackage{wasysym}
\usepackage{threeparttable}

\newcommand{\fullcircle}{\CIRCLE}
\newcommand{\emptycircle}{\Circle}
\newcommand{\halfcircle}{\LEFTcircle}

\theoremstyle{definition}
\newtheorem{definition}{Definition}
\newtheorem{theorem}{Theorem}
\newtheorem{lemma}{Lemma}

\begin{document}

\bstctlcite{UniPCH:BSTcontrol}
\title{UniPCH: A Unified Framework for Cross-Chain Payments via Payment Channel Hubs}
\author{
    Yuanming Shao~\orcidlink{0009-0009-3020-3494}, Yuming Feng~\orcidlink{0000-0001-8922-0496}, Weizhe Zhang~\orcidlink{0000-0003-4783-876X},~\IEEEmembership{Senior Member,~IEEE},
    Jianhuan Wang~\orcidlink{0000-0003-0031-2533}, Jian Chen~\orcidlink{0000-0002-7207-1751}, Song~Guo~\orcidlink{0000-0001-9831-2202},~\IEEEmembership{Fellow,~IEEE}, Zejun Wang~\orcidlink{0009-0009-0034-5459}, Bin Xiao~\orcidlink{0000-0003-4223-8220},~\IEEEmembership{Fellow,~IEEE}
    \thanks{
        Yuanming Shao is with the School of Cyberspace Science, Harbin Institute of Technology, Harbin 150001, China (e-mail: ymshao@stu.hit.edu.cn).
    }
    \thanks{
        Yuming Feng is with the Department of New Networks, Pengcheng Laboratory, Shenzhen 518055,
        China (e-mail: fengym@pcl.ac.cn).
    }
    \thanks{
        Weizhe Zhang is with the School of Computer Science and Technology,
        Harbin Institute of Technology (Shenzhen), Shenzhen 518055, China, and with the Department of New Networks, Pengcheng Laboratory, Shenzhen
        518055, China (e-mail: wzzhang@hit.edu.cn). (Corresponding authors: Weizhe Zhang and Yuming Feng.)
    }
    \thanks{
        Jianhuan Wang is with the Department of Computing,
        The Hong Kong Polytechnic University, Hung Hom, Hong Kong (e-mail:
        21122251r@connect.\allowbreak polyu.hk).
    }
    \thanks{
        Jian Chen is with The Hong Kong University of Science and Technology,
        Hong Kong (e-mail: kenchin@ust.hk).
    }
    \thanks{
        Song Guo is with the Department of Computer Science and Engineering, Hong Kong University of Science and Technology, Hong Kong (e-mail: songguo@cse.ust.hk).
    }
    \thanks{
        Zejun Wang is with the School of Cyberspace Science, Harbin Institute of Technology, Harbin 150001, China (e-mail: zejunwang@stu.hit.edu.cn).
    }
    \thanks{
        Bin Xiao is with the Department of Computing, The Hong Kong Polytechnic University, Hung Hom, Hong Kong (e-mail: b.xiao@polyu.edu.hk).
    }
    \thanks{This work has been submitted to IEEE for possible publication.
    Copyright may be transferred without notice, after which this version may no longer be accessible.}
}

\markboth{}{UniPCH: A Unified Framework for Cross-chain Payment}

\maketitle

\begin{abstract}
    Payment channel hubs (PCHs) enable private cross-chain payments, but existing protocols often depend on native transaction structures and channel update authorization. Supporting new ledgers or execution models requires corresponding channel realizations and may require protocol adaptation and security analysis. We present UniPCH, a two-layer framework separating native channel execution from private cross-chain payment logic. Its lower layer, UniPC, defines common channel operations and security requirements, enabling the same upper payment protocol to coordinate different channel realizations. UniPC separates conditional task authorization from confirmation of the complete channel state. Valid task evidence can be reused across updates, while both participants confirm the complete updated state. At the upper layer, we develop selective re-signing blind conditional signatures (S-RBCS) to construct a four-phase protocol over UniPC coordinating each realization's authorization tasks across hubs. We provide realizations for UTXO ledgers and account ledgers that support smart contracts. Our analysis establishes local UniPC security, griefing resistance, economic atomicity, and session unlinkability. We evaluate each phase's running time and bandwidth using a Go prototype across network conditions and system parameters. Results show modest computational overhead from the common interfaces. In the account configuration, selective refresh generates 95.4\% fewer adaptor pre-signatures per payment than full refresh.
\end{abstract}


\section{Introduction}

The continued expansion of blockchain ecosystems has led to the emergence of numerous blockchain systems and a rapidly growing user population. As applications and users increasingly span multiple blockchains, demand for cross-chain asset transfers continues to grow. However, differences in system architectures, transaction rules, consensus mechanisms, and execution environments make secure and efficient transfers difficult~\cite{CrossChannel,Sok,InteropSurvey}. Numerous cross-chain transfer protocols improve interoperability~\cite{ParaSwap,PACE,MAP,DataInterop,IvyCross}, but payment privacy has received less attention. Publicly observable addresses and transaction records can be correlated across blockchains to infer payment relationships and user behavior~\cite{Sok,zkCross,CAPE}. Protocols requiring frequent on-chain interactions also incur confirmation latency and transaction fees~\cite{Ultraviolet,Duplex}. These limitations motivate a cross-chain payment system combining interoperability, security, user privacy, and efficiency.

\begin{figure}[ht]
    \centering
    \includegraphics[width=\columnwidth]{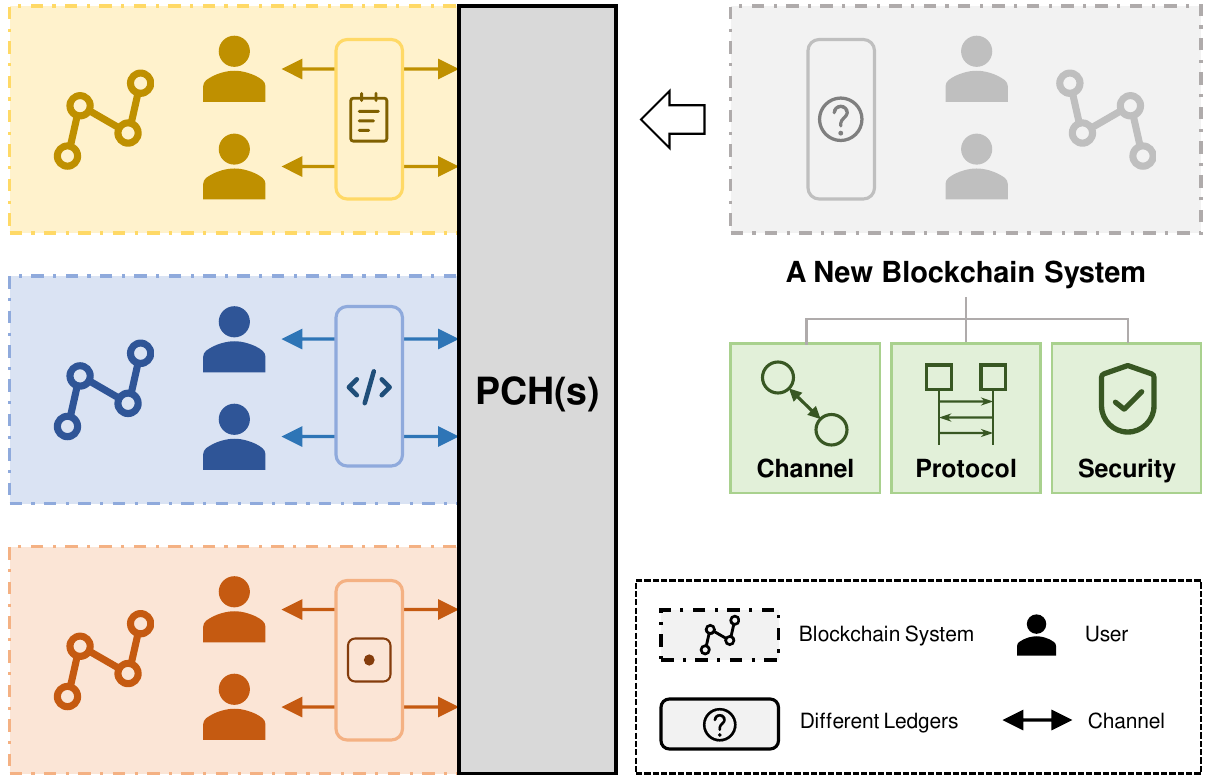}
    \caption{Dependence of existing PCH schemes on specialized channel, protocol, and security designs.}
    \label{fig:intro}
\end{figure}

Payment channel hubs (PCHs) enable efficient and private off-chain payments by allowing users to transact through a hub~\cite{OffchainSurvey}. Early schemes based on hubs include Bolt~\cite{Bolt}, Perun~\cite{Perun}, NOCUST~\cite{NOCUST}, and Teechain~\cite{Teechain}. They support efficient, scalable, or private off-chain payments. Their reliance on specific blockchain functions, smart contracts, or trusted hardware limits interoperability across heterogeneous blockchain environments. Subsequent protocols, including A$^2$L~\cite{A2L}, A$^2$L$^+$~\cite{A2LPlus}, and BlindHub~\cite{BlindHub}, reduce platform dependence and improve interoperability. They also improve unlinkability and atomicity through different cryptographic constructions. P2C2T~\cite{P2C2T} and EC2P~\cite{EC2P} further extend private payments with PCH to cross-chain settings. XPCH~\cite{XPCH} connects multiple PCHs serving different blockchains, allowing each hub to manage its local channels while distributing payment processing and liquidity requirements across hubs.

However, existing PCH protocols still have limitations. Some protocols specify channel operations through abstract functionalities without detailing their realizations for different ledger models. Other protocols build their payment flows around particular ledger models. Their payment logic depends on native transaction structures and ledger-specific procedures for authorizing, updating, and settling channel states. The same operation can require different authorizations across channel realizations. State updates can also affect whether existing conditional authorizations remain valid. For example, the channel instantiations presented by BlindHub~\cite{BlindHub} and XPCH~\cite{XPCH} target only unspent transaction output (UTXO) ledgers. EC2P~\cite{EC2P} uses separate protocols for each transfer direction between Turing-complete (TC) and non-Turing-complete (NTC) blockchains. These concrete channel designs do not provide a common workflow for authorizing updates across ledger models. As illustrated in Fig.~\ref{fig:intro}, extending a ledger-specific payment design to a different ledger or execution model requires a new channel realization and protocol adaptation. The protocol also needs to preserve payment security and privacy across these different channel mechanisms. Coupling payment logic to native channel mechanisms makes protocols difficult to reuse across heterogeneous blockchains and limits practical interoperability.

Addressing this gap requires identifying the behavior shared by heterogeneous channel implementations and coordinating them through a consistent private payment flow. This raises two challenges.

\textbf{CH1.} \emph{How can heterogeneous payment channels be used uniformly?}
Payment channels on different blockchains rely on different transaction structures, validation rules, and on-chain capabilities. These differences affect channel authorization, state updates, and settlement. For example, replacing a UTXO commitment transaction changes the signing messages for pending claims. An account realization can preserve authorizations for unchanged items across updates. Supporting heterogeneous channels uniformly requires shared channel behavior and security guarantees expressed independently of ledger-specific mechanisms.

\textbf{CH2.} \emph{How can heterogeneous payment channels be coordinated for secure cross-chain payments?}
A cross-chain payment may traverse channels with different authorization and settlement mechanisms. Each channel enforces local rules, while the payment requires security and privacy across the entire path. A general protocol needs to establish these properties for payments over heterogeneous channels using common channel guarantees.

To address these challenges, we propose UniPCH, a unified two-layer framework for cross-chain payments via PCHs. UniPCH adopts a multi-hub architecture similar to that of XPCH~\cite{XPCH}. Multiple PCHs connect the two layers: at the lower layer, each PCH manages channels local to its ledger through UniPC; at the upper layer, the PCHs participate in the UniPCH protocol to form a complete cross-chain payment path. For \textbf{CH1}, UniPC defines common channel behavior and security requirements for cross-chain payments. Each realization maps the common channel operations to native transactions or contract calls and derives the authorization tasks required by each update. UniPC's authorization lifecycle separates task evidence collection from confirmation of the complete candidate state. Valid task evidence is reused across updates, and fresh evidence is generated for the remaining tasks. Both parties then confirm the complete updated state. We provide realizations for UTXO ledgers and account ledgers that support smart contracts, and prove that both satisfy UniPC's security requirements. Other blockchain environments can be incorporated through conforming realizations. For \textbf{CH2}, we develop selective re-signing blind conditional signatures (S-RBCS) to coordinate conditional authorization across hubs for the tasks generated by each realization. UniPCH invokes only the UniPC interfaces along a path through multiple PCHs to coordinate heterogeneous realizations without interpreting native transactions. With each channel satisfying UniPC's security requirements, we use S-RBCS security to establish economic atomicity and session unlinkability for heterogeneous cross-chain payments.

The main contributions of this paper are summarized as follows.

\begin{itemize}
    \item We propose UniPCH, a unified two-layer framework that separates channel execution on individual ledgers from private cross-chain payment logic. This separation broadens interoperability for payments between users on heterogeneous blockchains. The same payment protocol supports their channels despite differences in native update, authorization, and settlement mechanisms.

    \item We develop UniPC, a payment channel abstraction that specifies common update semantics and security requirements for different ledger realizations. Each realization derives authorization tasks, and UniPC separates their evidence from confirmation of the complete channel state. Selective refresh retains valid evidence across updates, reducing repeated evidence generation and exchange. New evidence is generated for the remaining tasks. We provide conforming UTXO and account realizations.
    
    \item We develop S-RBCS to coordinate conditional authorization across hubs for the tasks generated by each realization. Building on S-RBCS, we construct a payment protocol that follows UniPC's evidence retention and refresh rules. It coordinates homogeneous and mixed ledger paths without interpreting native transactions. The protocol preserves payment security and user privacy and reduces unnecessary sequential interactions.

    \item We formalize component and system security. We prove that the UTXO and account realizations satisfy UniPC's security requirements and establish the upper protocol's security through the common abstraction. We implement a Go prototype and evaluate it against representative PCH schemes under different network conditions and system parameters. The experiments quantify the overhead of the common interfaces, the costs of both realizations, and the benefits of selective evidence refresh.
\end{itemize}

\section{Related Work}
\label{sec:related-work}

\noindent\textbf{Payment Channel.}
A payment channel (PC)~\cite{Duplex,Lightning} is an off-chain payment technique in which two parties open and fund a channel on chain, then conduct multiple payments by updating their channel balances off chain. Channels opened by different users form a payment channel network (PCN)~\cite{PCN1,Lightning,Sprites}, which routes payments between users without a direct channel. In a payment channel hub (PCH), users open channels with a hub that coordinates their transfers. This structure shortens payment paths, simplifies routing, and supports mechanisms that hide relationships between payers and payees. Based on platform dependence and cross-chain support, we group PCH schemes into platform-specific, platform-portable, and cross-chain PCHs. Table~\ref{tab:pch-comparison} compares representative schemes.

\begin{table*}[ht]
\centering
\caption{Comparison with PCH Schemes}
\label{tab:pch-comparison}
\renewcommand{\arraystretch}{1.08}
\footnotesize
\begin{threeparttable}
\setlength{\tabcolsep}{2.3pt}
\begin{tabular}{|c|c|c|c|c|c|c|c|}
\hline
\textbf{Scheme} & \textbf{Interoperability} & \textbf{Atomicity} & \textbf{Unlinkability} & \textbf{Griefing Resistance} & \textbf{Multi-Hubs} & \textbf{Unified PC Security} & \textbf{Unified Payment Protocol} \\
\hline
Bolt~\cite{Bolt} & \emptycircle~(Zcash) & \fullcircle & \fullcircle & \emptycircle & \fullcircle & \emptycircle & \emptycircle \\
Perun~\cite{Perun} & \emptycircle~(Ethereum) & \fullcircle & \emptycircle & \fullcircle & \emptycircle & \emptycircle & \emptycircle \\
NOCUST~\cite{NOCUST} & \emptycircle~(Ethereum) & \fullcircle & \emptycircle & \fullcircle & \emptycircle & \emptycircle & \emptycircle \\
Teechain~\cite{Teechain} & \halfcircle~(TEE) & \fullcircle & \fullcircle & \fullcircle & \fullcircle & \emptycircle & \emptycircle \\
\hline
TumbleBit~\cite{TumbleBit} & \halfcircle~(HTLC) & \fullcircle & \fullcircle & \emptycircle & \emptycircle & \emptycircle & \emptycircle \\
A$^2$L$(^+)$~\cite{A2L,A2LPlus}  & \fullcircle & \fullcircle & \fullcircle & \fullcircle & \emptycircle & \emptycircle & \emptycircle \\
BlindHub~\cite{BlindHub} & \fullcircle & \fullcircle & \fullcircle & \fullcircle & \emptycircle & \emptycircle & \emptycircle \\
\hline
P2C2T~\cite{P2C2T} & \fullcircle & \fullcircle & \fullcircle & \fullcircle & \emptycircle & \emptycircle & \emptycircle \\
EC2P~\cite{EC2P} & \fullcircle & \fullcircle & \fullcircle & \fullcircle & \emptycircle & \emptycircle & \emptycircle \\
XPCH~\cite{XPCH} & \fullcircle & \fullcircle & \fullcircle & \fullcircle & \fullcircle & \emptycircle & \emptycircle \\
\textbf{UniPCH} & \fullcircle & \fullcircle & \fullcircle & \fullcircle & \fullcircle & \fullcircle & \fullcircle \\
\hline
\end{tabular}
\begin{tablenotes}[flushleft]
\footnotesize
\item[] \mbox{\fullcircle\hspace{0.2em}/\hspace{0.2em}\halfcircle\hspace{0.2em}/\hspace{0.2em}\emptycircle}\hspace{0.5em}denote full support, conditional support, and a feature not provided in the cited design, respectively. Interoperability denotes payment support across blockchains or execution environments. Parentheses indicate the specific ledger, script, or hardware dependency. Multi-Hubs indicates whether one payment involves multiple hubs or nodes. Unified PC Security denotes common channel update and authorization interfaces and security requirements for realizations with different native update and settlement mechanisms. Unified Payment Protocol denotes a single payment flow across these realizations, with its security analysis based on the common channel guarantees.
\end{tablenotes}
\end{threeparttable}
\end{table*}

\noindent\textbf{Platform-Specific PCH.}
Platform-specific PCHs rely on capabilities provided by particular ledgers or execution environments. Bolt~\cite{Bolt} combines blind signatures, commitments, and ledger support to construct anonymous bidirectional channels. Perun~\cite{Perun} builds virtual payment hubs from smart contract state channels and provides a formal security proof, but requires a Turing-complete platform. NOCUST~\cite{NOCUST} protects commit-chain users through an on-chain contract, periodic state commitments, and challenge-response dispute mechanisms. Teechain~\cite{Teechain} protects channel states using trusted execution environments. These dependencies bind channel states, transactions, and dispute handling to specialized ledger or hardware capabilities, limiting portability and interoperability across heterogeneous blockchain environments.

\noindent\textbf{Platform-Portable PCH.}
Platform-portable PCHs use broadly available cryptographic or scripting capabilities to reduce requirements on the underlying blockchain and improve portability and interoperability. TumbleBit~\cite{TumbleBit} uses puzzle promises and puzzle solving to enable unlinkable payments through an untrusted intermediary with limited scripting. A$^2$L~\cite{A2L} replaces payment conditions based on hashes with adaptor signatures and formalizes anonymous atomic locks. A$^2$L$^+$~\cite{A2LPlus} revisits this construction and its security using linear-only encryption. BlindHub~\cite{BlindHub} combines blind and adaptor signatures to support payments with variable amounts on ledgers compatible with Bitcoin.

A$^2$L models channel operations through an ideal functionality for generalized channels. A$^2$L$^+$ develops blind conditional signatures for anonymous payments. BlindHub realizes BlindChannel only for ledgers compatible with Bitcoin. Their cryptographic mechanisms can be adapted for cross-chain payments, but supporting a new ledger model still requires a corresponding channel realization and the integration of update, authorization, and settlement procedures specific to that ledger. These works do not specify a common authorization workflow for channel realizations with different update and settlement mechanisms.

\noindent\textbf{Cross-Chain PCH.}
Cross-chain PCHs enable private payments through hubs across different ledgers. P2C2T~\cite{P2C2T} combines threshold anonymous atomic locks with verifiable timed discrete logarithm (VTD) schemes to provide atomicity, unlinkability, and indistinguishability with minimal ledger functionality. EC2P~\cite{EC2P} uses separate NTC-to-TC and TC-to-NTC protocols for payments between Turing-complete and non-Turing-complete blockchains. XPCH~\cite{XPCH} connects multiple PCHs through adaptor payment channels and uses RA$^2$L to propagate payment conditions across hubs. It supports multiple bidirectional payments per epoch and distributes liquidity and processing across hubs, reducing individual hub load.

The P2C2T and XPCH prototypes use UTXO ledgers such as Bitcoin and Litecoin. EC2P covers TC and NTC blockchains using separate protocols for each direction.

Supporting a new ledger model requires a new channel realization and may also require protocol adaptation and corresponding security analysis. As summarized in Table~\ref{tab:pch-comparison}, UniPCH specifies common operations and security requirements for channel realizations with different update, authorization, and settlement mechanisms. Each realization implements these operations using native mechanisms and satisfies the stated security requirements. The upper payment protocol coordinates these realizations through the same UniPC interfaces. Its security analysis relies on the common channel guarantees. New conforming realizations can be integrated without redesigning the upper payment flow.

\noindent\textbf{Comparison with XPCH.}
UniPCH and XPCH~\cite{XPCH} both adopt a multi-hub architecture, and UniPCH uses the re-signing adaptor signature (RAS) abstraction and re-signing blind conditional signature (RBCS) construction introduced in XPCH. XPCH's APC derives signing messages from commitment and descendant transactions. Replacing a commitment changes the messages of dependent claim transactions, and RAS generates replacement pre-signatures for them. With $k$ active transfers, an update therefore requires $O(k)$ adaptor pre-signature generations in this construction. UniPC supports this lifecycle and realizations that retain evidence for unchanged items. Its account realization uses fresh nonces to distinguish items and keeps their authorization messages independent of the global state version. Adding a transfer requires only $O(1)$ new adaptor pre-signatures, while both participants still confirm the complete updated state. UniPC specifies common security requirements for these different authorization lifecycles. We prove that the UTXO and account realizations satisfy these requirements. S-RBCS propagates blinded payment conditions across hubs and supports both retained and refreshed authorizations. We establish UniPCH's economic atomicity and session unlinkability using S-RBCS security and the security of channel realizations that meet UniPC's security requirements. Compared with XPCH, this design broadens interoperability for payments between users on heterogeneous blockchains. Conforming realizations can reuse the upper payment flow and its security analysis.

\section{Preliminaries}
\label{sec:preliminaries}

This section fixes the notation used throughout the paper and reviews the cryptographic abstractions required by UniPCH.

\noindent\textbf{Notation.}
Let $\lambda$ be the security parameter. Let $\mathbb{G}$ be a cyclic group of prime order $q$ with generator $G$, and let $\mathbb{Z}_q^*=\mathbb{Z}_q\setminus\{0\}$. We write $a\xleftarrow{\$}\mathcal{S}$ for uniform sampling from a finite set $\mathcal{S}$, $H(\cdot)$ for a collision-resistant hash function, and $\Vert$ for concatenation. We use $\Pi_{\mathrm{Sig}}$ and $\Pi_{\mathrm{ZK}}$ to denote the state confirmation signature scheme and zero-knowledge proof system, respectively. We assume that $\Pi_{\mathrm{Sig}}$ is unforgeable under chosen-message attacks.

\noindent\textbf{Blind Randomizable Signatures.}
Let $\Pi_{\mathrm{BRS}}$ denote a blind randomizable signature scheme~\cite{PS} that supports key generation, blind signing, unblinding, verification, and signature randomization. During blind issuance, the user commits to a message and proves through $\Pi_{\mathrm{ZK}}$ that the commitment is well formed. The signer returns a blind signature, which the user unblinds and randomizes into a signature valid for the message. UniPCH uses this primitive to issue a credential used once in each payment. We rely on the blindness, unforgeability, and correctness under unblinding and randomization of $\Pi_{\mathrm{BRS}}$.

\noindent\textbf{Adaptor and Re-Signing Adaptor Signatures.}
Let $\Pi_{\mathrm{AS}}$ denote an adaptor signature scheme~\cite{AdaptorSignature}. Given a signing key $sk$, message $m$, and statement $Y$, $\mathsf{PreSign}(sk,m,Y)$ produces a pre-signature $\widetilde{\sigma}$. $\mathsf{PreVerify}$ verifies the pre-signature, $\mathsf{Adapt}$ completes it using witness $y$, and $\mathsf{Extract}$ recovers $y$ from $\widetilde{\sigma}$, the completed signature $\sigma$, and $Y$.

Let $\Pi_{\mathrm{RAS}}$ denote the re-signing adaptor signature abstraction of XPCH~\cite{XPCH}, which extends $\Pi_{\mathrm{AS}}$ with $\mathsf{RSign}$. When an update changes the ledger message from $m$ to $m'$, $\mathsf{RSign}$ uses the previously used signing key and statement to generate the replacement pre-signature $\widetilde{\sigma}'$ for $m'$. UniPCH uses $\Pi_{\mathrm{AS}}$ and $\Pi_{\mathrm{RAS}}$ to bind conditional transfer completion to payment conditions and refresh evidence when ledger-defined messages change. We rely on their correctness, unforgeability, and witness extractability.

\noindent\textbf{Linear-Only Encryption.}
A linear-only encryption scheme $\Pi_{\mathrm{LOE}}=(\mathsf{KeyGen},\mathsf{Enc},\mathsf{Dec},\mathsf{RandEnc})$ encrypts witnesses in $\mathbb{Z}_q$ and supports public scalar multiplication~\cite{A2LPlus}. For $Y=yG$, $C\leftarrow\mathsf{Enc}(pk_{\mathrm{LOE}},y)$, and $\alpha\xleftarrow{\$}\mathbb{Z}_q^*$, anyone can compute $(Y',C')=(\alpha Y,\mathsf{RandEnc}(C,\alpha))$. The resulting pair encrypts $\alpha y$ without revealing $y$ or $\alpha$. A proof $\pi_C$ shows that $Y$ and $C$ contain the same witness. UniPCH uses this property to randomize encrypted payment conditions across hubs. We rely on the correctness and confidentiality under randomization of $\Pi_{\mathrm{LOE}}$. The proof system $\Pi_{\mathrm{ZK}}$ for $\pi_C$ is complete, sound, and zero knowledge. Following A$^2$L$^+$~\cite{A2LPlus}, we prove Theorem~\ref{thm:srbcs-construction} in the idealized LOE model. Ciphertext operations are modeled through oracle access, with handles tracking affine combinations of witnesses. Our implementation uses CL linearly homomorphic encryption~\cite{CL} based on DDH.

\noindent\textbf{Re-Signing Blind Conditional Signatures.}
Let $\Pi_{\mathrm{RBCS}}$ denote the RBCS construction of XPCH~\cite{XPCH}. It combines blind randomizable credentials, LOE ciphertexts, and re-signing adaptor signatures to coordinate linked conditional authorizations among $S$, $R$, and their participating hubs. It provides setup and \emph{near}/\emph{far} variants of registration, promise, and solving. In the near variant, an endpoint communicates directly with the coordinating hub; in the far variant, it communicates through an intermediate hub. It satisfies correctness, blindness, unlockability, and unforgeability. These properties ensure valid outputs in honest executions, endpoint unlinkability, witness release across linked authorizations, and resistance to unauthorized completion, respectively.

UniPCH uses $\Pi_{\mathrm{RBCS}}$ as the basis for condition propagation across hubs. Building on it, we develop selective re-signing blind conditional signatures (S-RBCS), which bind this propagation to the authorization tasks generated by UniPC realizations and refresh the evidence required by each realization.

\section{System Overview}
\label{sec:overview}

This section presents the UniPCH execution environment, framework with two layers, payment flow, threat model, and security goals.

\subsection{Environment}

\noindent\textbf{System Model.}
The system comprises a set of blockchains $\mathcal{B}=\{B_1,\ldots,B_n\}$, a set of PCHs $\mathcal{H}=\{H_1,\ldots,H_n\}$, and a set of users $\mathcal{U}=\bigcup_{i=1}^{n}\mathcal{U}_i$. The collection $\mathcal{B}$ is potentially heterogeneous. Each blockchain $B_i$ is associated with a designated hub $H_i$, which serves the users in $\mathcal{U}_i$. Each user $u\in\mathcal{U}_i$ opens and funds a payment channel $c_{u,H_i}$ on $B_i$ for off-chain transfers with $H_i$. All hubs agree in advance on a common blockchain $B_h\in\mathcal{B}$ for their mutual channels. Every pair of distinct hubs establishes and funds a bilateral channel $c_{H_i,H_j}$ on $B_h$, forming a fully connected overlay.

For a payment from sender $S\in\mathcal{U}_i$ to receiver $R\in\mathcal{U}_j$, their access hubs are $H_S=H_i$ and $H_R=H_j$, respectively. Before execution, $S$ and $R$ jointly select a coordinating hub $H_C\in\mathcal{H}$ at random, yielding the logical path $(S,H_S,H_C,H_R,R)$. If $H_C$ is directly connected to $S$ or $R$, then $H_C=H_S$ or $H_C=H_R$, respectively, and the corresponding edge between hubs is omitted. If $H_C$ is directly connected to both endpoints, then $H_C=H_S=H_R$ and the path reduces to a single-hub payment. For the full path with three distinct hubs, the four channels are established on $(B_i,B_h,B_h,B_j)$ in order. The endpoint channels use their respective local blockchains, and both channels between hubs use the common blockchain $B_h$. When $B_i\neq B_j$, the path realizes a cross-chain payment. 

Throughout, \emph{payment} denotes the complete operation from $S$ to $R$, \emph{transfer} a value movement on one UniPC channel, and \emph{transaction} a native ledger object used by a UniPC realization.

\noindent\textbf{System Assumptions.} UniPCH relies on the following assumptions.

\noindent(1) Payments are executed in epochs, each divided into the ordered protocol phases described below. Payments assigned to the same epoch execute the corresponding phase during the same time interval. All participants have access to a synchronized clock with sufficient granularity to identify the current epoch and phase. Updates to the same UniPC channel are processed sequentially, with each update based on the latest installed state.

\noindent(2) All payments within an epoch are of equal value under the agreed exchange rates. The amounts on the participating blockchains are fixed before execution and remain unchanged throughout the protocol.

\noindent(3) Internal communication channels are authenticated and confidential, while the direct communication channel between $S$ and $R$ is anonymous.

\noindent(4) The underlying blockchains correctly process valid transactions or contract calls and confirm them within their assumed delay bounds.

\subsection{Our Framework}

\noindent\textbf{Architecture with Two Layers.} UniPCH separates native channel execution from cross-chain payment coordination. As shown in Fig.~\ref{fig:architecture}, PCHs bridge the two layers by managing user channels and channels between hubs through UniPC at the lower layer. At the upper layer, the PCHs coordinate cross-chain payments across these channels. User channels use their local blockchains, whereas all channels between hubs use the common blockchain $B_h$.

\begin{figure}[ht]
\centering
\includegraphics[width=\linewidth]{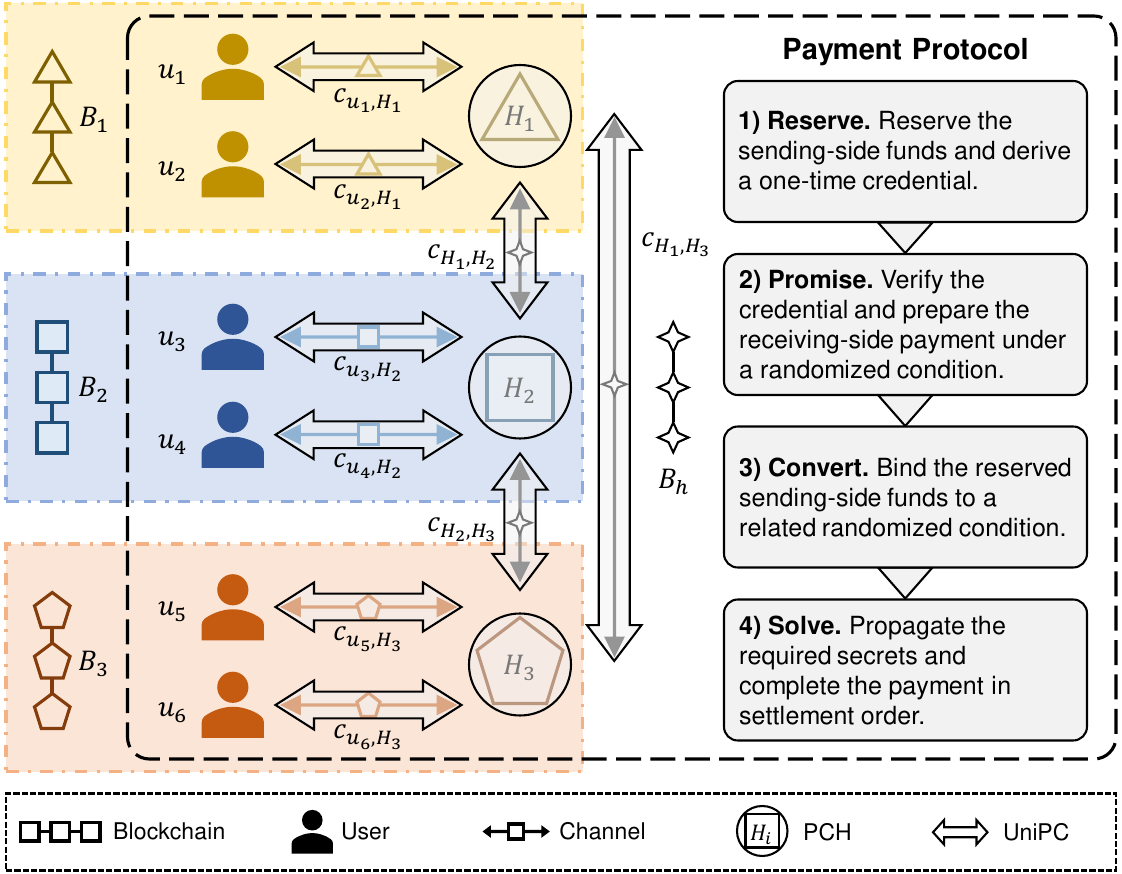}
\caption{Overview of UniPCH.}
\label{fig:architecture}
\end{figure}

\noindent\textbf{UniPC.} UniPC defines common state effects and security requirements for different channel realizations. Its six operations cover fund reservation, conditional prepayment, conversion, claiming, and refunds. This workflow is structurally similar to XPCH's APC~\cite{XPCH}, which implements it using UTXO transactions. UniPC's authorization lifecycle separates task evidence from confirmation of the complete candidate state. Each realization maps candidate updates to native transactions or contract calls and derives authorization tasks and settlement objects. The parties verify all required evidence before confirming the state and install the update after obtaining both state confirmation signatures. The upper protocol coordinates heterogeneous channels through these interfaces without interpreting native transactions. Section~\ref{sec:unipc} formalizes UniPC and provides conforming realizations for UTXO ledgers and account ledgers with smart contracts.

\noindent\textbf{UniPCH.} In UniPCH, multiple PCHs connect the local payment domains of different blockchains into a complete payment path. Each PCH manages its user channels and channels between hubs through the corresponding UniPC realizations. For a payment from $S$ to $R$, $H_S$ and $H_R$ connect the two users to the path, while $H_C$ coordinates the sender and receiver sides. When $H_S=H_C$ or $H_R=H_C$, the corresponding step between hubs is omitted.

As shown in Fig.~\ref{fig:architecture}, $\mathsf{Reserve}$ lets $S$ obtain a credential from $H_C$ backed by sender reservations. After these reservations are jointly confirmed, $S$ sends the credential to $R$. In $\mathsf{Promise}$, $R$ presents it to $H_C$ to request one prepayment. After the corresponding receiver-side updates are jointly confirmed, $R$ forwards the associated encrypted payment condition to $S$. In $\mathsf{Convert}$, $S$ uses this information to convert the sender reservations into conditional prepayments. In $\mathsf{Solve}$, $H_C$ first claims funds on the sender side. The required secrets are then recovered, transformed, and propagated to complete the remaining claims in order. Randomization of the credential and encrypted payment condition hides the association between $S$ and $R$. Each realization derives its own authorization and refresh tasks within this flow.

\subsection{Security}
\noindent\textbf{Threat Model.} We consider the following threats to UniPCH.

\noindent(1) \emph{Channel state attack.} A malicious channel participant may attempt to install an invalid, incompletely authorized, not jointly confirmed, or obsolete state, or cause conflicting terminal outcomes to be accepted.

\noindent(2) \emph{Fund theft attack.} A malicious user or hub may attempt to obtain an incoming transfer without fulfilling the linked outgoing transfer.

\noindent(3) \emph{Griefing attack.} A malicious receiver or receiver access hub may request multiple prepayments from an honest $H_C$ without distinct confirmed sender reservations, locking excessive hub funds.

\noindent(4) \emph{Linkage attack.} For privacy, $S$ and $R$ are honest, and participating hubs are honest-but-curious and may collude by sharing their protocol views and observations from public ledgers to link $S$ and $R$.

\noindent\textbf{Security Goals.} UniPCH addresses these threats through the following goals.

\noindent(1) \emph{Local UniPC security.} An honest participant accepts only channel states that are valid, completely authorized, and jointly confirmed. Such states preserve value, prevent obsolete states from becoming terminal, and enforce exclusive terminal outcomes.

\noindent(2) \emph{Economic atomicity.} An honest endpoint either completes its intended transfer or recovers its locked funds, while every claimed outgoing transfer of an honest intermediary is backed by an incoming transfer that has been claimed or remains claimable.

\noindent(3) \emph{Griefing resistance.} Every receiver prepayment installed by an honest $H_C$ is backed by a distinct $\mathsf{Reserve}$ session whose sender reservation updates of the agreed amount have been jointly confirmed, and each credential authorizes at most one such prepayment.

\noindent(4) \emph{Session unlinkability.} Within an epoch, colluding honest-but-curious hubs cannot distinguish how honest senders and receivers are paired among payments with the same $H_C$ and agreed amount, except with negligible advantage.

\section{UniPC}
\label{sec:unipc}

This section presents UniPC, a payment channel abstraction independent of ledger type. We define its common channel state, six operations, authorization and confirmation interfaces, and security requirements. We also provide conforming UTXO and account realizations. Fig.~\ref{fig:unipc-state-authorization} shows how the operations produce candidate states and how the authorization lifecycle installs a selected candidate.

\subsection{State and Operations}
\label{sec:unipc-state}

We first define the common channel state maintained by UniPC. At version $\nu$, a channel between parties $A$ and $B$ is represented by
\begin{equation}
 \begin{gathered}
    \Sigma_\nu=(\nu,\phi_A,\phi_B,\mathcal{I}),\\
    x=(\mathsf{id},\mathsf{type},\mathsf{payer},v,\tau),
 \end{gathered}
    \label{eq:unipc-state}
\end{equation}
where $\phi_A$ and $\phi_B$ are the spendable balances and $\mathcal{I}$ is the collection of active items. An item $x$ represents either locked funds or a conditional prepayment. $\mathsf{id}$ uniquely identifies the item, $\mathsf{payer}$ identifies its payer, and the other channel participant is its payee. The fields $v$ and $\tau$ denote the reserved value and refund deadline.

\begin{figure}[ht]
    \centering
    \includegraphics[width=\columnwidth]{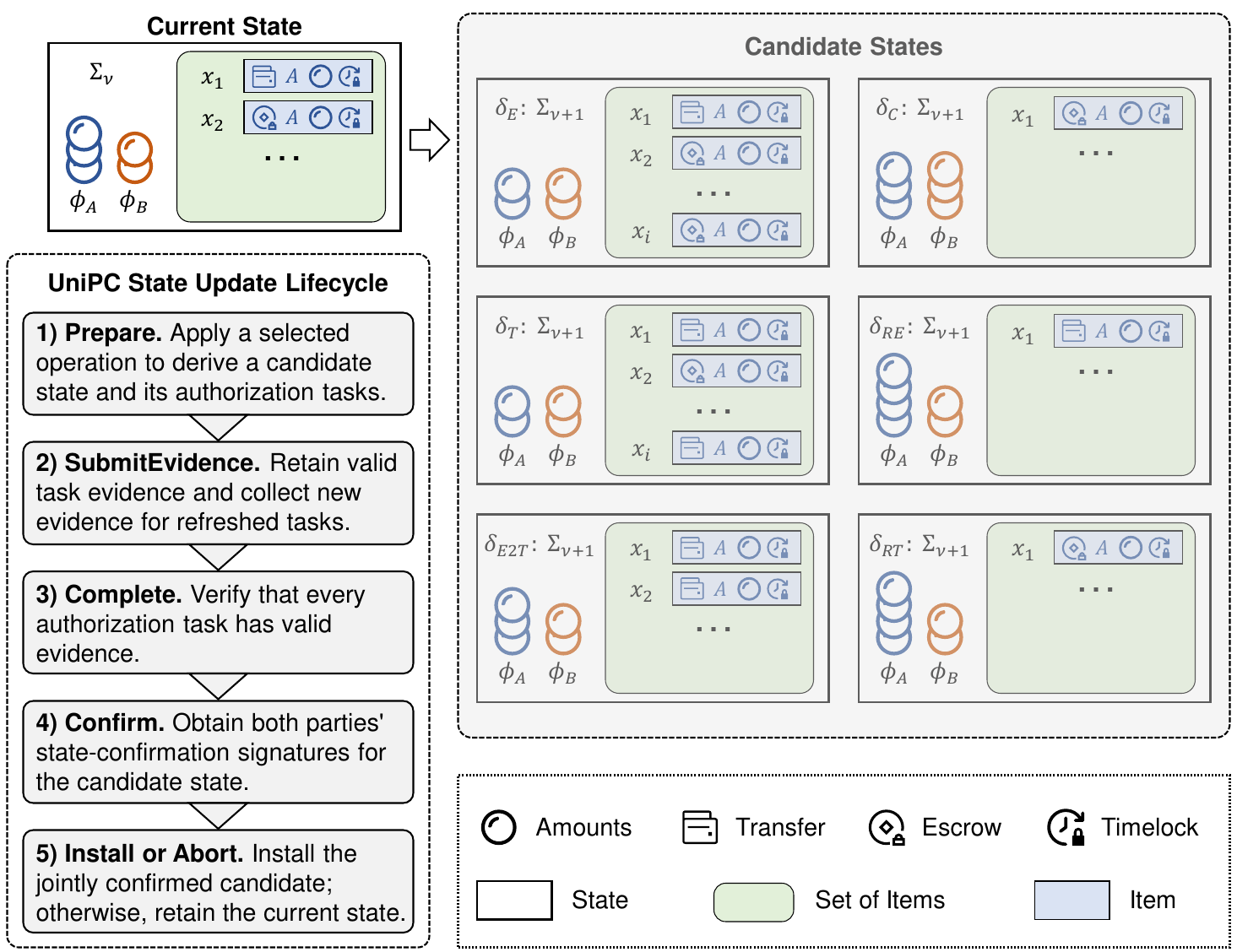}
    \caption{UniPC state transitions and the authorization lifecycle.}
    \label{fig:unipc-state-authorization}
    \end{figure}

To open a channel, $A$ and $B$ lock deposits $\psi_A$ and $\psi_B$ on the underlying ledger. The channel is initialized as
\begin{equation}
    \Sigma_0=(0,\psi_A,\psi_B,\emptyset).
    \label{eq:unipc-initial-state}
\end{equation}
Thus, the initial spendable balances equal the deposits and the active item set is empty.

We describe a transfer of amount $\mu$ from $A$ to $B$ on one channel. The reverse direction is symmetric. Both escrow and conditional transfer items reserve $\mu$, representing locked funds and a prepayment, respectively. We write $e(\mu)$ and $t(\mu)$ as shorthand for their complete item tuples. We denote $\mathsf{Escrow}$, $\mathsf{Transfer}$, $\mathsf{Escrow2Transfer}$, $\mathsf{ClaimTransfer}$, $\mathsf{RefundEscrow}$, and $\mathsf{RefundTransfer}$ by $\delta_{\mathsf{E}}$, $\delta_{\mathsf{T}}$, $\delta_{\mathsf{E2T}}$, $\delta_{\mathsf{C}}$, $\delta_{\mathsf{RE}}$, and $\delta_{\mathsf{RT}}$, respectively. Their state effects are shown below, where $\mathcal{I}$ contains the unaffected items in each row and omitted components remain unchanged:

{\small
\begin{equation}
\begin{aligned}
 \delta_{\mathsf{E}}:\;&
 (\phi_A,\mathcal{I})
 \rightarrow(\phi_A-\mu,\mathcal{I}\cup\{e(\mu)\}),\\
 \delta_{\mathsf{T}}:\;&
 (\phi_A,\phi_B,\mathcal{I})
 \rightarrow(\phi_A-\mu,\phi_B,\mathcal{I}\cup\{t(\mu)\}),\\
 \delta_{\mathsf{E2T}}:\;&
 (\phi_A,\phi_B,\mathcal{I}\cup\{e(\mu)\})
 \rightarrow(\phi_A,\phi_B,\mathcal{I}\cup\{t(\mu)\}),\\
 \delta_{\mathsf{C}}:\;&
 (\phi_A,\phi_B,\mathcal{I}\cup\{t(\mu)\})
 \rightarrow(\phi_A,\phi_B+\mu,\mathcal{I}),\\
 \delta_{\mathsf{RE}}:\;&
 (\phi_A,\mathcal{I}\cup\{e(\mu)\})
 \rightarrow(\phi_A+\mu,\mathcal{I}),\\
 \delta_{\mathsf{RT}}:\;&
 (\phi_A,\phi_B,\mathcal{I}\cup\{t(\mu)\})
 \rightarrow(\phi_A+\mu,\phi_B,\mathcal{I}).
\end{aligned}
\label{eq:unipc-transitions}
\end{equation}
}

$\mathsf{Escrow}$ receives $(\mu,\tau)$, deducts $\mu$ from the payer, and creates an escrow item with an identifier generated by the realization. $\mathsf{Transfer}$ receives $(\mu,\tau,Y)$, similarly deducts $\mu$, and creates a conditional transfer item together with its authorization task for statement $Y$. $\mathsf{Escrow2Transfer}$ receives an escrow identifier, a new deadline, and $Y$, and replaces the escrow with a transfer while preserving its reserved value. Given a witness $y$ for the statement in the corresponding task, $\mathsf{ClaimTransfer}$ removes the transfer and credits $\mu$ to the payee. After the refund deadline, $\mathsf{RefundEscrow}$ and $\mathsf{RefundTransfer}$ remove the corresponding item and return $\mu$ to the payer.

An operation is admissible only if the payer has sufficient balance, each referenced item has the expected type and payer, and the applicable condition or deadline is satisfied. In particular, $\mathsf{ClaimTransfer}$ requires a valid witness for $Y$, whereas both refund operations require the local ledger time to have reached $\tau$. Applying an admissible operation increments $\nu$ and updates $\mathcal{I}$, producing a candidate state. A claimed or refunded item is removed and cannot be consumed again.

Every state reachable from $\Sigma_0$ through admissible operations satisfies the conservation invariant
\begin{equation}
    \phi_A+\phi_B+\sum_{x\in\mathcal{I}}v_x
    =\psi_A+\psi_B.
    \label{eq:unipc-conservation}
\end{equation}
A realization may introduce fields and public parameters specific to its ledger in the native encoding, but it must preserve the operation effects in \eqref{eq:unipc-transitions} and the invariant in \eqref{eq:unipc-conservation}. The protocol in the upper layer can use the same six operations regardless of how each realization encodes and executes them.

\subsection{Authorization and State Installation}
\label{sec:unipc-authorization}

The transition rules above are used to construct an admissible candidate state. The parties then collect and verify the task evidence required by the state update and complete the update after confirming the candidate. UniPC uses an authorization and confirmation lifecycle to determine whether the candidate is installed or discarded. We represent an update by the record
\begin{equation}
    \mathcal{R}_{\nu}
    =(\Sigma_{\nu},\mathcal{P}_{\nu},\mathcal{E}_{\nu},\Gamma_{\nu}),
    \label{eq:unipc-record}
\end{equation}
where $\Sigma_{\nu}$ is the candidate state at version $\nu$, $\mathcal{P}_{\nu}$ is the complete set of authorization tasks generated by the realization, $\mathcal{E}_{\nu}$ records the verified evidence for those tasks, and $\Gamma_{\nu}=(\gamma_{A,\nu},\gamma_{B,\nu})$ is the state confirmation signature pair.

Here, a task $\eta\in\mathcal{P}_{\nu}$ is represented as
\begin{equation}
    \eta=(\mathsf{signer},m,Y),
    \label{eq:unipc-task}
\end{equation}
where $\mathsf{signer}$ identifies the participant that provides authorization, $m$ is the authorization message defined by the ledger realization, and $Y$ is the condition statement for an adaptor signature task. The concrete realization defines $m$ so that it binds the relevant channel context, authorization role, and every item or settlement field whose change must invalidate the evidence. Each $e\in\mathcal{E}_{\nu}$ is a verified task evidence object. When checking retained or newly submitted evidence against a task, the interface evaluates
\begin{equation}
    \mathsf{VfyEv}(\eta,e)=1.
    \label{eq:unipc-verify-evidence}
\end{equation}
The lifecycle consists of $\mathsf{Prepare}$, $\mathsf{SubmitEvidence}$, $\mathsf{Complete}$, $\mathsf{Confirm}$, and $\mathsf{Install}$. $\mathsf{Abort}$ leaves the local installed record unchanged. A participant retains any candidate it has confirmed and resumes it when retrying. Jointly confirmed candidates follow the realization's settlement rules.

\noindent\textbf{1) $\mathsf{Prepare}$.}
Given the current installed state $\Sigma$, an operation $\mathsf{op}$, and its inputs $\theta$, $\mathsf{Prepare}$ produces:
\begin{equation}
    (\Sigma',\mathcal{P}')
    \leftarrow
    \mathsf{Prepare}(\Sigma,\mathsf{op},\theta),
    \label{eq:unipc-prepare}
\end{equation}
where $\mathsf{op}$ identifies one of the six UniPC operations and $\theta$ contains the inputs required by that operation, such as the amount, item identifier, deadline, condition statement, and any inputs required by the realization. The resulting $\Sigma'$ is the candidate state, and $\mathcal{P}'$ is the complete set of tasks required to authorize it.

After $\mathcal{P}'$ is returned, it is compared with the verified evidence $\mathcal{E}$ stored in the current installed record. The tasks already satisfied by existing evidence form
\begin{equation}
 \mathcal{P}_{\mathrm{keep}}
 =\{\eta\in\mathcal{P}':\exists e\in\mathcal{E},
   \mathsf{VfyEv}(\eta,e)=1\}.
\label{eq:unipc-refresh}
\end{equation}
The remaining tasks that require new evidence form $\mathcal{P}_{\mathrm{refresh}}=\mathcal{P}'\setminus\mathcal{P}_{\mathrm{keep}}$.

\noindent\textbf{2) $\mathsf{SubmitEvidence}$.}
Evidence satisfying $\mathcal{P}_{\mathrm{keep}}$ is retained from the current installed record, while the corresponding channel participants generate new evidence for $\mathcal{P}_{\mathrm{refresh}}$. Each channel participant submits evidence for the tasks requiring its authorization. Each submitted evidence object is first checked using \eqref{eq:unipc-verify-evidence}. Evidence that passes verification is accepted as verified evidence.

\noindent\textbf{3) $\mathsf{Complete}$.}
Let $\mathcal{E}'$ denote the set containing the accepted new evidence and the retained evidence for $\mathcal{P}_{\mathrm{keep}}$. UniPC uses the following predicate to check whether every task in $\mathcal{P}'$ has valid evidence:
{\small
\begin{equation}
 \mathsf{Comp}(\mathcal{P}',\mathcal{E}')
 =\bigwedge_{\eta\in\mathcal{P}'}
   \left(\exists e\in\mathcal{E}':
   \mathsf{VfyEv}(\eta,e)=1\right).
 \label{eq:unipc-complete}
\end{equation}
}
\noindent\textbf{4) $\mathsf{Confirm}$.}
An honest participant confirms at most one successor candidate for each installed version. After \eqref{eq:unipc-complete} holds in its local view, participant $i\in\{A,B\}$ signs the candidate state message $m'_{\mathsf{st}}=\mathsf{StateMsg}(\mathsf{dom}_c,\Sigma')$:
\begin{equation}
 \gamma'_i\leftarrow\mathsf{Sign}(sk_i,m'_{\mathsf{st}}),
 \qquad \Gamma'=(\gamma'_A,\gamma'_B).
\label{eq:unipc-confirm}
\end{equation}
The confirmation certificate is valid when
\begin{equation}
 \mathsf{VfyConf}(\Sigma',\Gamma')
 =\bigwedge_{i\in\{A,B\}}
 \left[\mathsf{Verify}(pk_i,m'_{\mathsf{st}},\gamma'_i)=1\right].
\label{eq:unipc-verify-confirm}
\end{equation}
\noindent\textbf{5) $\mathsf{Install}$.}
The candidate state satisfies the installation condition when
\begin{equation}
 \mathsf{Inst}(\Sigma',\mathcal{P}',\mathcal{E}',\Gamma')
 =\mathsf{Comp}(\mathcal{P}',\mathcal{E}')\land\mathsf{VfyConf}(\Sigma',\Gamma').
\label{eq:unipc-installable}
\end{equation}
The bilateral update protocol is summarized in Fig.~\ref{fig:unipc-update-protocol}.

\begin{figure}[t]
\centering
\footnotesize
\setlength{\tabcolsep}{3pt}
\renewcommand{\arraystretch}{1.08}
\begin{tabularx}{\linewidth}{|r>{\raggedright\arraybackslash}X|}
\hline
\multicolumn{2}{|X|}{\textbf{UniPC authorization and state installation}}\\
\hline
\multicolumn{2}{|X|}{\textbf{Participants:} initiator $A$ and counterparty $B$.}\\
\multicolumn{2}{|X|}{\textbf{Input:} installed record $\mathcal R=(\Sigma,\mathcal P,\mathcal E,\Gamma)$, operation $\mathsf{op}$, and inputs $\theta$. Each party uses its state confirmation key pair $(sk_i,pk_i)$.}\\
\hline
1: & $A$: $(\Sigma',\mathcal P')\leftarrow\mathsf{Prepare}(\Sigma,\mathsf{op},\theta)$.\\
2: & $A$: Determine $\mathcal P_{\mathrm{keep}}$ and $\mathcal P_{\mathrm{refresh}}$ using \eqref{eq:unipc-refresh}, and retain $\mathcal E_{\mathrm{keep}}$.\\
3: & $A$: Generate $\mathcal E_A^{\mathrm{new}}$ for its assigned tasks.\\
4: & $A\rightarrow B$: Send the update request for $(\mathsf{op},\theta)$ with $\mathcal E_A^{\mathrm{new}}$.\\
5: & $B$: Reconstruct and check the candidate and its tasks. Determine the retained and refreshed tasks, and verify $\mathcal E_A^{\mathrm{new}}$ against the corresponding tasks using $\mathsf{VfyEv}$.\\
6: & $B$: Generate $\mathcal E_B^{\mathrm{new}}$ and assemble $\mathcal E'=\mathcal E_{\mathrm{keep}}\cup\mathcal E_A^{\mathrm{new}}\cup\mathcal E_B^{\mathrm{new}}$.\\
7: & $B$: Check $\mathsf{Comp}(\mathcal P',\mathcal E')=1$, then compute $\gamma'_B\leftarrow\mathsf{Sign}(sk_B,m'_{\mathsf{st}})$.\\
8: & $B\rightarrow A$: $(\mathcal E_B^{\mathrm{new}},\gamma'_B)$.\\
9: & $A$: Verify $\mathcal E_B^{\mathrm{new}}$ against its tasks and assemble $\mathcal E'$. Check $\mathsf{Comp}(\mathcal P',\mathcal E')=1$ and $\mathsf{Verify}(pk_B,m'_{\mathsf{st}},\gamma'_B)=1$.\\
10: & $A$: Compute $\gamma'_A\leftarrow\mathsf{Sign}(sk_A,m'_{\mathsf{st}})$.\\
11: & $A\rightarrow B$: $\gamma'_A$.\\
12: & Each party forms $\Gamma'=(\gamma'_A,\gamma'_B)$ and installs $\mathcal R'=(\Sigma',\mathcal P',\mathcal E',\Gamma')$ after checking $\mathsf{Inst}(\Sigma',\mathcal P',\mathcal E',\Gamma')=1$ locally.\\
\hline
\multicolumn{2}{|X|}{On a failed check or incomplete exchange, invoke $\mathsf{Abort}$ and retain the installed record and any locally confirmed candidate.}\\
\hline
\end{tabularx}
\caption{Authorization and state installation in UniPC.}
\label{fig:unipc-update-protocol}
\end{figure}

A realization conforms to UniPC when it implements the transitions in \eqref{eq:unipc-transitions}, preserves \eqref{eq:unipc-conservation}, and uses a fixed algorithm for deriving tasks that includes every participant authorization required by its native encoding. It installs a record after \eqref{eq:unipc-installable} evaluates to $1$ and maps $\mathsf{ClaimTransfer}$, $\mathsf{RefundEscrow}$, and $\mathsf{RefundTransfer}$ to terminal actions verifiable by the ledger. 

To incorporate a ledger with a different execution model, its realization implements native funding, state encoding, the six operation mappings, and task evidence generation and verification. The realization follows UniPC's confirmation and installation rules and uses native settlement mechanisms to enforce refund deadlines and handle obsolete states. It provides the common UniPC interfaces to the payment protocol after satisfying these requirements and Definition~\ref{def:unipc-security}.

\begin{definition}[Security of UniPC]
\label{def:unipc-security}
A conforming UniPC realization is secure if it satisfies \emph{transition correctness}, \emph{authorization security}, \emph{state freshness}, and \emph{settlement consistency}.
\end{definition}

Transition correctness requires every installed state to follow an admissible operation and preserve the channel funds. Authorization security requires valid evidence for every task and confirmation of the complete state by both participants. State freshness prevents an obsolete state from becoming terminal when an honest participant follows the prescribed response. Settlement consistency excludes conflicting claim and refund outcomes for the same item. Appendix~\ref{sec:complete-unipc-security} gives their complete definitions.

\subsection{Realizations}
\label{sec:unipc-realizations}

As representative examples, we provide UniPC realizations for two widely used ledger architectures: the UTXO model used by Bitcoin~\cite{Bitcoin} and the account model used by Ethereum~\cite{Ethereum}.

\noindent\textbf{UTXO Realization.}
To open a channel, parties $A$ and $B$ contribute their deposits to a funding transaction $\mathsf{TX}_F$, whose output $o_F$ is locked by a 2-of-2 multisignature script under their funding keys. The parties use $o_F$ as the common input of a jointly signed initial commitment transaction and represent subsequent channel states by replacing this commitment transaction. At version $\nu$, the abstract state is encoded as
{\small
\begin{equation}
\begin{gathered}
 \mathsf{TX}_C^{\nu}:o_F\longrightarrow
 \bigl(o_A(\phi_A),o_B(\phi_B),\{o_x(v_x,\mathsf{scr}_x)\}_{x\in\mathcal I}\bigr),\\
 \mathsf{Enc}_{\mathrm{UTXO}}(\Sigma_{\nu})
 =\bigl(\mathsf{TX}_C^{\nu},\{\mathcal D_x\mid x\in\mathcal I\}\bigr),
\end{gathered}
\label{eq:unipc-utxo}
\end{equation}
}
where $\mathsf{TX}_C^{\nu}$ spends $o_F$. Its first two outputs $o_A$ and $o_B$ encode the spendable balances. Each $o_x(v_x,\mathsf{scr}_x)$ is an active item output of reserved value $v_x$ locked by script $\mathsf{scr}_x$, and $\mathcal D_x$ contains the descendant transactions implementing its permitted terminal actions. A reference to $o_x$ that is local to the realization supplies $\mathsf{id}_x$. The refund deadline $\tau_x$ is enforced by a timelocked refund path, through which the payer recovers an escrow or an unclaimed transfer after $\tau_x$. A transfer output additionally has a claim branch requiring the signatures associated with the payer and payee keys. The payer's signature on the descendant claim transaction is provided as an adaptor pre-signature under statement $Y$. After learning witness $y$, the payee completes it and spends the output together with its own signature.

The operations in \eqref{eq:unipc-transitions} are realized by modifying the commitment transaction as follows. $\mathsf{Escrow}$ deducts $\mu$ from the payer and appends an escrow output with deadline $\tau$. $\mathsf{Transfer}$ likewise deducts $\mu$ and appends a conditional transfer output with its deadline and claim and refund descendants. $\mathsf{Escrow2Transfer}$ replaces the referenced escrow output with a transfer output, preserves its reserved value, and updates $\tau$ to the new transfer deadline. $\mathsf{ClaimTransfer}$ removes the transfer output and credits the payee after the claim condition is satisfied, while $\mathsf{RefundEscrow}$ and $\mathsf{RefundTransfer}$ remove an expired item and return its value to the payer. In these admissible transitions, the balance and active item outputs satisfy \eqref{eq:unipc-conservation}.

For each state update, $\mathsf{Prepare}$ follows the rules above to construct the candidate commitment transaction $\mathsf{TX}_C'$ and its descendant transactions. For every active transfer $x$ paid by participant $i$, the realization derives an adaptor signature task whose signer is $i$, whose message $m_x$ is the signature hash of the descendant claim transaction, and whose statement is $Y_x$. Newly created tasks form $\mathcal{P}_{\mathrm{new}}$, while carried tasks whose claim transaction has changed form $\mathcal{P}_{\mathrm{affected}}$. Thus, the task set returned for the candidate is
\begin{equation}
 \mathcal{P}'_{\mathrm{UTXO}}
 =\mathcal{P}_{\mathrm{new}}
 \cup\mathcal{P}_{\mathrm{affected}}.
 \label{eq:unipc-utxo-tasks}
\end{equation}
Each participant generates an adaptor pre-signature for the tasks assigned to it. Because an unsettled claim transaction spends an output of $\mathsf{TX}_C'$, replacing the commitment changes the referenced transaction hash and may change the output index. Its signing message therefore changes, so every active transfer carried into the new commitment belongs to $\mathcal{P}_{\mathrm{affected}}$ and its payer needs to refresh the corresponding evidence. Therefore, $\mathcal{P}_{\mathrm{keep}}=\emptyset$ and $\mathcal{P}_{\mathrm{refresh}}=\mathcal{P}'_{\mathrm{UTXO}}$. $\mathsf{SubmitEvidence}$ collects the adaptor pre-signatures required by $\mathcal{P}_{\mathrm{refresh}}$. $\mathsf{Complete}$ checks that every task has valid evidence, and $\mathsf{Confirm}$ obtains both ordinary signatures on $\mathsf{TX}_C'$. $\mathsf{Install}$ adopts the jointly confirmed candidate. $\mathsf{Abort}$ retains the preceding installed commitment and any candidate already confirmed by the participant.

The concrete bilateral authorization and installation process is as follows. The initiator sends the update request and its task evidence to the counterparty. The counterparty reconstructs the candidate transactions, checks the input, output amounts, scripts, item indices, deadlines, and value conservation, and uses \eqref{eq:unipc-verify-evidence} to verify every submitted pre-signature. It then generates the evidence for its own tasks. Once \eqref{eq:unipc-complete} holds, the counterparty signs $\mathsf{TX}_C'$ and returns its task evidence and state confirmation signature. The initiator verifies the response, signs the same $\mathsf{TX}_C'$, and returns its state confirmation signature. These two ordinary signatures form $\Gamma'$. After both signatures have been verified, the parties install the same candidate state. Failure to obtain the required evidence or either signature invokes $\mathsf{Abort}$, which leaves the local installed commitment unchanged.

The state encoding, operation mappings, authorization tasks, and installation procedure above constitute a conforming UniPC realization. For a transfer output, the shared item output prevents its claim and refund branches from both settling. Revocation keys and the mechanism for handling stale states prevent an obsolete commitment from replacing the latest jointly confirmed balances.

\begin{theorem}[Security of the UTXO Realization]
\label{thm:utxo-realization}
Assuming the unforgeability of $\Pi_{\mathrm{Sig}}$ and $\Pi_{\mathrm{RAS}}$, and assuming that the underlying UTXO ledger correctly validates transactions and confirms valid transactions within the assumed delay bound, the UTXO realization above satisfies Definition~\ref{def:unipc-security}. The proof is provided in Appendix~\ref{sec:proof-utxo-realization}.
\end{theorem}

\noindent\textbf{Account Realization.}
One party registers participants $A$ and $B$ with a channel contract to open a channel, after which they deposit $\psi_A$ and $\psi_B$ into the contract. The parties maintain jointly authorized states off-chain, while the contract holds the deposits and processes a signed state only when on-chain settlement is required. At version $\nu$, the state is encoded as
\begin{equation}
\begin{gathered}
    \mathsf{Enc}_{\mathrm{SC}}(\Sigma_\nu)
    =(\nu,\phi_A,\phi_B,C_{\mathcal{I}}),\\
    h_{\Sigma}
    =H(\mathsf{cid}\Vert\nu\Vert\phi_A\Vert\phi_B
    \Vert C_{\mathcal{I}}),
\end{gathered}
    \label{eq:unipc-account}
\end{equation}
where $C_{\mathcal I}$ is a commitment defined by the realization and derived from the abstract active item collection $\mathcal I$. It commits to the identifiers of all active items and is realized as either a Merkle root with membership proofs or a hash of an ordered list with item indices. The representation is fixed when the channel opens. An ordered list is simple to maintain for small active item sets but requires submitting the full list for on-chain settlement. A Merkle root keeps the submitted state compact, while claims and refunds require membership proofs. The digest $h_\Sigma$ binds the channel identifier, version, balances, and complete active item set and is used to derive the state confirmation message. For an item $x$, the account realization derives its identifier and authorization message as
\begin{equation}
 m_x=\mathsf{id}_x=H(\mathsf{type}_x\Vert\mathsf{cid}\Vert\mathsf{payer}_x\Vert v_x\Vert\tau_x\Vert n_x),
\label{eq:unipc-account-item}
\end{equation}
where $n_x$ is a fresh nonce used by the realization when deriving $\mathsf{id}_x$, preventing otherwise identical items from sharing an identifier, while $\mathsf{payer}_x$ identifies the payer.

For on-chain settlement, the contract reconstructs $\mathsf{id}_x$ from the item fields and verifies its membership in the finalized state using a Merkle proof or the corresponding list index. It also checks that the item has not been consumed. For $\mathsf{ClaimTransfer}$, the payee uses the witness to complete the payer's adaptor pre-signature on $m_x$. The contract verifies the completed signature and pays $v_x$ to the payee. For $\mathsf{RefundEscrow}$ or $\mathsf{RefundTransfer}$, it checks that the refund deadline $\tau_x$ has passed and returns $v_x$ to the payer. Before transferring funds, the contract marks the item as consumed.

The operations in \eqref{eq:unipc-transitions} are realized by updating this off-chain record. Whenever an operation creates a new item, the realization samples a fresh $n_x$ before deriving its identifier. $\mathsf{Escrow}$ deducts $\mu$ from the payer and inserts an escrow item with deadline $\tau$. $\mathsf{Transfer}$ likewise deducts $\mu$ and inserts a conditional transfer item. $\mathsf{Escrow2Transfer}$ replaces the referenced escrow with a transfer item, preserves its value, and updates its deadline. After the claim authorization is completed, $\mathsf{ClaimTransfer}$ removes the transfer and credits the payee, while $\mathsf{RefundEscrow}$ and $\mathsf{RefundTransfer}$ remove an expired item and return its value to the payer. Each operation increments $\nu$ and recomputes $C_{\mathcal I}$, and the resulting balances and active items satisfy \eqref{eq:unipc-conservation}.

Let $\mathcal P$ denote the set of authorization tasks for the installed state. In $\mathsf{Prepare}$, each active transfer $x$ has one adaptor signature task $\eta_x$ assigned to its payer, with message $m_x$ and statement $Y_x$. Since $m_x$ excludes $\nu$ and $h_\Sigma$, an unchanged transfer retains its task when another item is updated. A $\mathsf{Transfer}$ or $\mathsf{Escrow2Transfer}$ operation adds only the task $\eta_x$ of the new transfer, and in this case, $\mathcal P_{\mathrm{keep}}=\mathcal P$ and $\mathcal P_{\mathrm{refresh}}=\{\eta_x\}$. $\mathsf{Escrow}$ and $\mathsf{RefundEscrow}$ leave the set of transfer tasks unchanged, with $\mathcal P_{\mathrm{keep}}=\mathcal P$ and $\mathcal P_{\mathrm{refresh}}=\emptyset$. $\mathsf{ClaimTransfer}$ and $\mathsf{RefundTransfer}$ remove the consumed transfer's task, with $\mathcal P_{\mathrm{keep}}=\mathcal P\setminus\{\eta_x\}$ and $\mathcal P_{\mathrm{refresh}}=\emptyset$.

$\mathsf{SubmitEvidence}$ collects the adaptor pre-signatures required by $\mathcal P_{\mathrm{refresh}}$. $\mathsf{Complete}$ checks the retained and newly submitted evidence against $\mathcal P'$. $\mathsf{Confirm}$ obtains fresh state confirmation signatures on $h_\Sigma'$. $\mathsf{Install}$ adopts the jointly confirmed state. $\mathsf{Abort}$ retains the preceding state.

The concrete bilateral authorization and installation process is as follows. The initiator submits the update request and any refreshed evidence assigned to it. The counterparty reconstructs the candidate and checks its contents. It verifies newly submitted evidence using \eqref{eq:unipc-verify-evidence}, retains the evidence for $\mathcal P_{\mathrm{keep}}$, and supplies any refreshed evidence assigned to it. After \eqref{eq:unipc-complete} holds, it signs the candidate state message and returns its new evidence and state confirmation signature. The initiator verifies the response, signs the same candidate, and returns its state confirmation signature. After both signatures have been verified, the parties install the same state. Failure to obtain both verified signatures invokes $\mathsf{Abort}$, which retains the preceding installed state.

The state encoding, operation mappings, authorization tasks, confirmation procedure, and contract settlement constitute a conforming UniPC realization. To settle on chain, either party submits a state together with both confirmation signatures. The contract verifies these signatures on $h_\Sigma$ and starts a challenge period. During this period, a jointly signed state with a higher version supersedes the current submission. After the challenge period, the contract finalizes the highest submitted state for claim and refund processing. An honest payee submits the latest jointly confirmed state early enough for the challenge period, state finalization, and confirmation of its $\mathsf{ClaimTransfer}$ transaction to complete before the refund deadline $\tau_x$.

\begin{theorem}[Security of the Account Realization]
\label{thm:account-realization}
Assuming the collision resistance of $H$, the unforgeability of $\Pi_{\mathrm{Sig}}$ and $\Pi_{\mathrm{AS}}$, and correct execution and timely confirmation of valid contract calls by the underlying account ledger, the account realization above satisfies Definition~\ref{def:unipc-security}. The proof is provided in Appendix~\ref{sec:proof-account-realization}.
\end{theorem}

\section{UniPCH}
\label{sec:unipch}

UniPCH coordinates private payments across multi-hub paths whose channels may use different UniPC realizations. The protocol invokes common UniPC operations, and every channel update follows the UniPC authorization and state installation lifecycle. This section presents path invocation, the S-RBCS binding between adjacent channels, the payment protocol with four phases, and failure handling.

\subsection{Payment Path and UniPC Invocation}
\label{sec:unipch-model}

For a payment from sender $S$ to receiver $R$, let $H_S$ and $H_R$ be their access hubs and let $H_C$ be the selected coordinating hub. The logical path is
\begin{equation}
    \mathcal{Q}=(S,H_S,H_C,H_R,R).
    \label{eq:unipch-path}
\end{equation}
All hubs are interconnected through payment channels implemented on the designated blockchain $B_h\in\mathcal{B}$. Since $S$ and $R$ may belong to different blockchains, a payment may span multiple heterogeneous blockchains. If $H_S=H_C$ or $H_R=H_C$, the corresponding edge between hubs is omitted. If both equalities hold, the protocol reduces to a single-hub payment. UniPCH executes payments in epochs, each divided into the ordered phases described below.

We use $\mathsf{Apply}(c,\mathsf{op},\theta)$ as shorthand for applying operation $\mathsf{op}$ with inputs $\theta$ to channel $c$ through the UniPC lifecycle. Let $b_{\mathrm{inst}}=\mathsf{Inst}(\Sigma',\mathcal{P}',\mathcal{E}',\Gamma')$. The operation installs the new record when $b_{\mathrm{inst}}=1$ and otherwise retains the preceding record:
\begin{equation}
 \mathsf{Apply}(c,\mathsf{op},\theta;\mathcal{R}_{\nu})=
 \begin{cases}
  \mathcal{R}_{\nu+1}, & b_{\mathrm{inst}}=1,\\
  \mathcal{R}_{\nu}, & \text{otherwise}.
 \end{cases}
 \label{eq:unipch-apply}
\end{equation}
\subsection{S-RBCS}
\label{sec:unipch-crypto}

UniPCH defines selective re-signing blind conditional signatures (S-RBCS), a stateful conditional authorization scheme. It binds anonymous credentials and randomized conditions to UniPC tasks and selectively refreshes affected evidence. Its syntax is $\Pi_{\mathrm{S\text{-}RBCS}}=(\mathsf{Setup},\allowbreak\mathsf{Reserve},\allowbreak\mathsf{Promise},\allowbreak\mathsf{Convert},\allowbreak\mathsf{Solve})$. $\mathsf{Setup}$ initializes the public parameters and keys. The four interactive procedures reserve the sending path and issue a one-time credential, prepare the receiving path, convert the sender-side reservations, and complete the linked authorizations, respectively.

S-RBCS integrates Near and Far operations into the same ordered phases. Near operations connect an endpoint directly to $H_C$. Far operations reach $H_C$ through an access hub. It invokes the corresponding UniPC operations on each existing edge and omits an edge between hubs when its adjacent roles coincide. The set of authorization tasks affected by a state change is determined by each channel realization. For example, a UTXO realization returns tasks affected by transaction replacement, while an account realization with stable identifiers normally returns only a newly created transfer task. Both follow the common UniPC lifecycle for final state confirmation.

\begin{definition}[Security of S-RBCS]
\label{def:srbcs-security}
An S-RBCS scheme $\Pi_{\mathrm{S\text{-}RBCS}}$ is secure if it satisfies \emph{correctness}, \emph{task unforgeability}, \emph{unlockability}, and \emph{blindness}.
\end{definition}

Correctness ensures that honest executions produce valid credentials, linked conditions, and the required task evidence. Task unforgeability excludes complete signatures on unauthorized messages. For an honest coordinating hub, it also prevents completed outgoing authorizations from outnumbering successful incoming solves. Unlockability ensures that completing an adaptor signature reveals the witness required by the linked task. Blindness prevents credentials, randomized conditions, and task evidence from revealing the association between endpoint sessions with the same public inputs and phase schedule. The complete definitions of these properties are provided in Appendix~\ref{sec:complete-srbcs-security}.

We construct $\Pi_{\mathrm{S\text{-}RBCS}}^{*}$ using $\Pi_{\mathrm{BRS}}$, $\Pi_{\mathrm{LOE}}$, $\Pi_{\mathrm{AS}}$, $\Pi_{\mathrm{RAS}}$, $\Pi_{\mathrm{ZK}}$, and the UniPC authorization interface. It propagates linked conditions across hubs and processes the authorization tasks derived by UniPC realizations under the common evidence retention and refresh rules. The required adaptor signature capabilities depend on whether updates change existing task messages. For example, our account realization preserves task messages across updates and uses only $\Pi_{\mathrm{AS}}$ for task authorization, lowering the cryptographic requirements for integration into UniPCH. Our UTXO realization uses $\Pi_{\mathrm{RAS}}$ to refresh evidence after transaction replacement.

\noindent\textbf{1) $\mathsf{Setup}$.} On input $1^\lambda$, $\mathsf{Setup}$ initializes the public parameters and keys for $\Pi_{\mathrm{BRS}}$, $\Pi_{\mathrm{LOE}}$, $\Pi_{\mathrm{ZK}}$, and the adaptor signature schemes required by the UniPC realizations.

We give the task authorization rule. For a task $\eta=(\mathsf{signer},m,Y)$ with $Y=xG$, the signer produces an adaptor pre-signature $\widetilde{\sigma}\leftarrow\mathsf{PreSign}(sk,m,Y)$ as evidence.
For each candidate state, $\Pi_{\mathrm{S\text{-}RBCS}}^{*}$ processes the tasks returned by its UniPC realization as
\begin{equation}
 e'_{\eta}\leftarrow
 \begin{cases}
  e_{\eta}, & \eta\in\mathcal{P}_{\mathrm{keep}},\\
  \widetilde{\sigma}^{\mathrm{new}}_{\eta}, & \eta\in\mathcal{P}_{\mathrm{refresh}}.
 \end{cases}
 \label{eq:srbcs-refresh}
\end{equation}
New transfer tasks use $\mathsf{PreSign}$. When an existing task's authorization message changes, $\mathsf{RSign}$ refreshes its evidence. All evidence must satisfy \eqref{eq:unipc-verify-evidence} before the participants confirm the candidate state through UniPC.

\noindent\textbf{2) $\mathsf{Reserve}$.} The sender $S$ samples and commits to a fresh $\mathsf{tid}$. The sender-side channels invoke $\mathsf{Escrow}$ and process the resulting task evidence. After verifying the commitment proof and the reservation evidence adjacent to $H_C$, the hub returns a blind response $\sigma^*$, which $S$ unblinds and randomizes:
\begin{equation}
\begin{gathered}
 \mathsf{tid}\xleftarrow{\$}\{0,1\}^{\lambda},\quad
 (cm,d,\pi_{cm})\leftarrow\mathsf{Commit}(\mathsf{tid}),\\
 \sigma^*\leftarrow\mathsf{BlindSign}(sk_{H_C}^{\mathrm{BRS}},cm),\\
 \sigma_{\mathsf{tid}}\leftarrow\mathsf{Randomize}(\mathsf{Unblind}(\sigma^*,d)).
\end{gathered}
\label{eq:unipch-credential}
\end{equation}
Each $\mathsf{Reserve}$ session yields one credential, released to $R$ after the sender-side reservations are jointly confirmed and usable for one receiver prepayment.

\noindent\textbf{3) $\mathsf{Promise}$.} $R$ presents the credential and $\mathsf{tid}$ to request a prepayment from $H_C$. After verification, $H_C$ records $\mathsf{tid}$, samples $x_0\in\mathbb{Z}_q^*$, sets $Y_0=x_0G$, and verifiably encrypts $x_0$ as $(C_0,\pi_0)$. The receiver-side channels invoke $\mathsf{Transfer}(Y_0)$ and submit the required task evidence. $R$ verifies $\pi_0$ and its adjacent task evidence, then confirms its candidate state. Fig.~\ref{fig:srbcs-promise-protocol} details the procedure for $H_R\neq H_C$.

\begin{figure}[t]
\centering
\footnotesize
\setlength{\tabcolsep}{3pt}
\renewcommand{\arraystretch}{1.08}
\begin{tabularx}{\linewidth}{|r>{\raggedright\arraybackslash}X|}
\hline
\multicolumn{2}{|>{\raggedright\arraybackslash}X|}{\textbf{Promise in $\Pi_{\mathrm{S\text{-}RBCS}}^{*}$}}\\
\hline
\multicolumn{2}{|>{\raggedright\arraybackslash}X|}{\textbf{Participants:} $H_C$, $H_R$, and $R$, with $H_R\neq H_C$.}\\
\multicolumn{2}{|>{\raggedright\arraybackslash}X|}{\textbf{Input:} credential $(\mathsf{tid},\sigma_{\mathsf{tid}})$ held by $R$, agreed amount and deadlines, channel records, and keys from $\mathsf{Setup}$.}\\
\hline
1: & $R\rightarrow H_R\rightarrow H_C$: $(\mathsf{tid},\sigma_{\mathsf{tid}})$.\\
2: & $H_C$: Verify the credential under $pk_{H_C}^{\mathrm{BRS}}$ and check that $\mathsf{tid}$ is unused. Record its use.\\
3: & $H_C$: Sample $x_0\xleftarrow{\$}\mathbb Z_q^*$ and set $Y_0=x_0G$.\\
4: & $H_C$: Compute $C_0=\mathsf{Enc}(pk_{\mathrm{LOE}},x_0;\rho_0)$ with fresh randomness $\rho_0$.\\
5: & $H_C$: $\pi_0\leftarrow\Pi_{\mathrm{ZK}}.\mathsf{Prove}((pk_{\mathrm{LOE}},Y_0,C_0),(x_0,\rho_0))$.\\
6: & $H_C$: Prepare $\mathsf{Transfer}(Y_0)$ on $c_{H_C,H_R}$. Retain valid task evidence and generate the required new evidence.\\
7: & $H_C\rightarrow H_R$: $(Y_0,C_0,\pi_0)$ with the update request and task evidence.\\
8: & $H_R$: Check the received task evidence using $\mathsf{VfyEv}$. Prepare the $\mathsf{Transfer}(Y_0)$ update on $c_{H_R,R}$ and its task evidence.\\
9: & $H_C,H_R$: Complete the update on $c_{H_C,H_R}$ through UniPC.\\
10: & $H_R\rightarrow R$: After $c_{H_C,H_R}$ is jointly confirmed, send $(Y_0,C_0,\pi_0)$ with the update request and task evidence.\\
11: & $R$: Check $\Pi_{\mathrm{ZK}}.\mathsf{Verify}((pk_{\mathrm{LOE}},Y_0,C_0),\pi_0)=1$.\\
12: & $R$: Verify the task evidence against the candidate tasks using $\mathsf{VfyEv}$ before confirming the candidate state.\\
13: & $H_R,R$: Complete the channel update through UniPC.\\
\hline
\multicolumn{2}{|>{\raggedright\arraybackslash}X|}{\textbf{Output:} jointly confirmed receiver-side transfers whose new tasks use $Y_0$, and the verified pair $(Y_0,C_0)$ held by $R$.}\\
\multicolumn{2}{|>{\raggedright\arraybackslash}X|}{Failed checks invoke the existing abort and recovery rules.}\\
\hline
\end{tabularx}
\caption{Promise in the S-RBCS construction for $H_R\neq H_C$.}
\label{fig:srbcs-promise-protocol}
\end{figure}

\noindent\textbf{4) $\mathsf{Convert}$.} $R$ randomizes the verified condition and ciphertext using a nonzero $\alpha_R$, yielding $x_1=\alpha_Rx_0$. $S$ then applies a nonzero $\alpha_S$, yielding $x_2=\alpha_Sx_1$:
\begin{equation}
\begin{gathered}
 (Y_1,C_1)=(\alpha_RY_0,\mathsf{RandEnc}(C_0,\alpha_R)),\\
 (Y_2,C_2)=(\alpha_SY_1,\mathsf{RandEnc}(C_1,\alpha_S)).
\end{gathered}
\label{eq:unipch-randomization}
\end{equation}
The receiver-side channels retain $Y_0$, while the sender-side channels invoke $\mathsf{Escrow2Transfer}(Y_2)$. Only $(Y_1,C_1)$ passes directly from $R$ to $S$. Each candidate follows the task authorization rule before confirmation and installation through UniPC.

\noindent\textbf{5) $\mathsf{Solve}$.} $H_C$ decrypts $C_2$ to obtain $x_2$. For a local task with $Y=xG$, the payee completes its pre-signature using $x$, and the payer extracts $x$:
\begin{equation}
\begin{gathered}
 \sigma\leftarrow\mathsf{Adapt}(\widetilde{\sigma},x),\quad
 x\leftarrow\mathsf{Extract}(\widetilde{\sigma},\sigma,Y).
\end{gathered}
\label{eq:unipch-adaptor}
\end{equation}
$S$ derives $x_1=\alpha_S^{-1}x_2$, and $R$ derives $x_0=\alpha_R^{-1}x_1$. The witnesses complete the linked authorizations and propagate in settlement order as the channels invoke $\mathsf{ClaimTransfer}$.

\begin{theorem}[Security of the S-RBCS Construction]
\label{thm:srbcs-construction}
Assume that the one-more discrete logarithm (OMDL) problem is hard in $\mathbb{G}$ and that $H$ is collision resistant. Suppose that $\Pi_{\mathrm{BRS}}$ is correct under unblinding and randomization and satisfies blindness, and that $\Pi_{\mathrm{LOE}}$ is correct and confidential under randomization. Suppose further that $\Pi_{\mathrm{AS}}$ and $\Pi_{\mathrm{RAS}}$ satisfy correctness, unforgeability, and witness extractability for the task evidence queries prescribed by the construction, and that $\Pi_{\mathrm{ZK}}$ is complete, sound, and zero knowledge. In the LOE model, assume that independent item authorizations under the same signing key use distinct authorization messages, and that task evidence and completed claims are checked against their corresponding UniPC tasks. Then $\Pi_{\mathrm{S\text{-}RBCS}}^{*}$ satisfies Definition~\ref{def:srbcs-security}. The proof is provided in Appendix~\ref{sec:security-srbcs}.
\end{theorem}

\subsection{Payment Protocol}
\label{sec:unipch-protocol}

The UniPCH payment protocol invokes the corresponding procedure of $\Pi_{\mathrm{S\text{-}RBCS}}^{*}$ in four phases: $\mathsf{Reserve}$, $\mathsf{Promise}$, $\mathsf{Convert}$, and $\mathsf{Solve}$. Each existing edge along $\mathcal{Q}$ applies the required UniPC operation through $\mathsf{Apply}$, while edges between coincident hub roles are omitted. Under S-RBCS, when an honest hub's outgoing transfer is claimed, its incoming transfer has been claimed or remains claimable. Fig.~\ref{fig:payment-flow} illustrates a payment using a UTXO UniPC on $c_{S,H_S}$ and account UniPCs on the other three channels. In each channel, $e_i$ denotes newly generated task evidence under $Y_i$, for $i\in\{0,2\}$.

\begin{figure}[ht]
\centering
\includegraphics[width=\columnwidth]{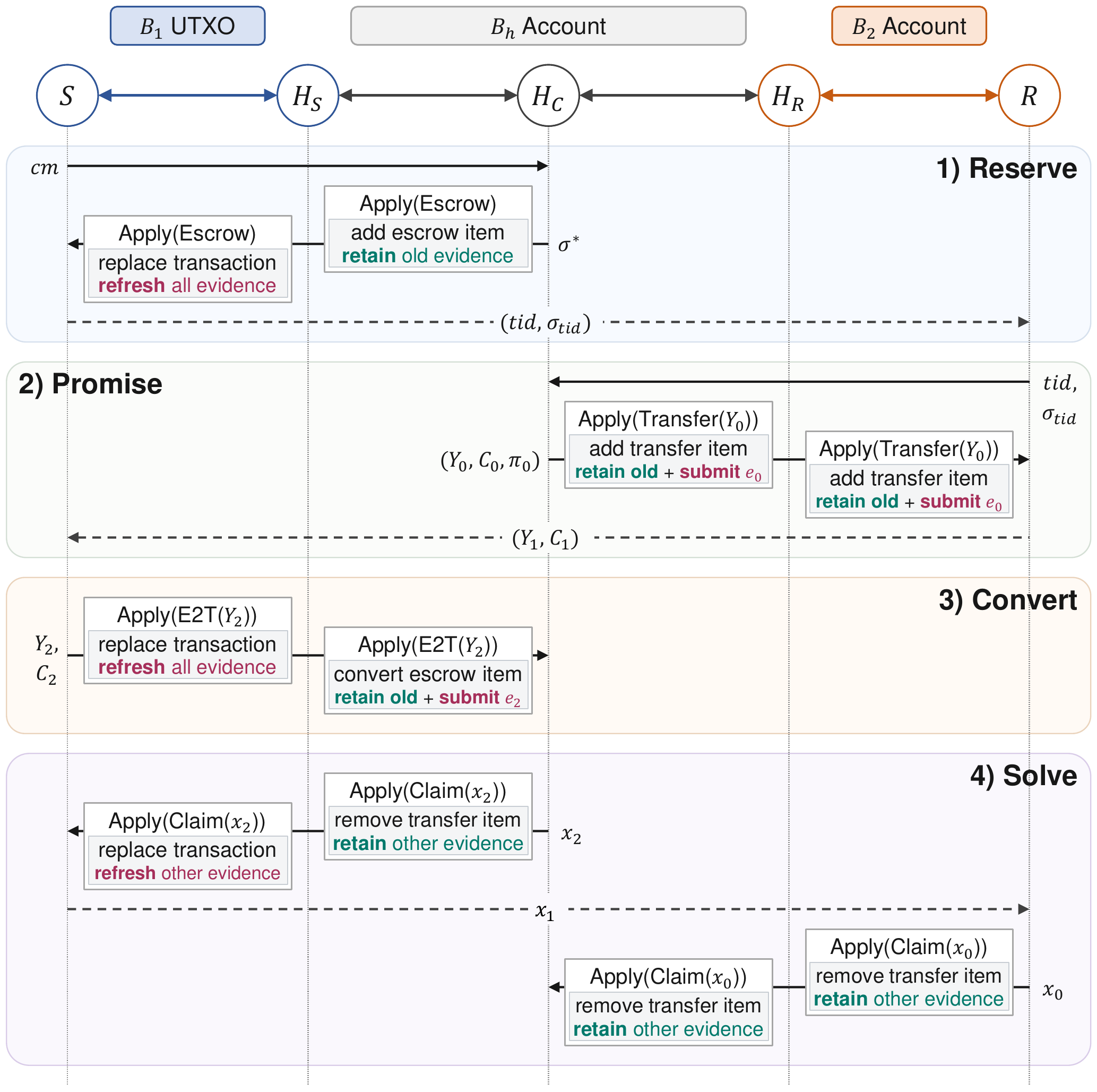}
\caption{Execution of the UniPCH Payment Protocol across Heterogeneous Ledgers.}
\label{fig:payment-flow}
\end{figure}

\noindent\textbf{1) $\mathsf{Reserve}$.}
$S$ creates $(\mathsf{tid},cm,\pi_{cm})$ and initiates an $\mathsf{Escrow}$ update on $c_{S,H_S}$. If $H_S\neq H_C$, $H_S$ prepares the corresponding $\mathsf{Escrow}$ on $c_{H_S,H_C}$. After verifying the commitment proof and its adjacent reservation evidence, $H_C$ returns a blind signature response along the same path. The associated reservation updates continue through their UniPC authorization and confirmation lifecycles. Once the reservations on the sender side are jointly confirmed, $S$ derives $\sigma_{\mathsf{tid}}$ and sends $(\mathsf{tid},\sigma_{\mathsf{tid}})$ directly to $R$.

\noindent\textbf{2) $\mathsf{Promise}$.}
$R$ presents $(\mathsf{tid},\sigma_{\mathsf{tid}})$ to $H_C$ through $H_R$. After verifying the credential, $H_C$ atomically checks whether $\mathsf{tid}$ has already been used and records it. It aborts the payment if a previous use is detected. For a newly recorded $\mathsf{tid}$, $H_C$ generates $(x_0,Y_0,C_0,\pi_0)$. When $H_R=H_C$, $H_C$ applies $\mathsf{Transfer}$ directly on $c_{H_C,R}$. When $H_R\neq H_C$, $H_C$ first applies $\mathsf{Transfer}$ on $c_{H_C,H_R}$. After verifying the corresponding adaptor pre-signature, $H_R$ applies $\mathsf{Transfer}$ with the same $Y_0$ on $c_{H_R,R}$. In either case, $R$ verifies $\pi_0$ and the adaptor pre-signature for its transfer before confirming its local update. Once all existing transfers on the receiver side are jointly confirmed, $R$ computes $(Y_1,C_1)$ and sends $(\mathsf{tid},Y_1,C_1)$ directly to $S$.

\noindent\textbf{3) $\mathsf{Convert}$.}
$S$ computes $(Y_2,C_2)$, sends $C_2$ to $H_C$ along the sender path, and initiates $\mathsf{Escrow2Transfer}$ on $c_{S,H_S}$. If $H_S\neq H_C$, $H_S$ waits until the update on $c_{S,H_S}$ is jointly confirmed. It then coordinates the same operation for the associated escrow on $c_{H_S,H_C}$. Each candidate is authorized, confirmed, and installed through its own $\mathsf{Apply}$. At completion, every sender transfer is conditioned on $Y_2$, and every receiver transfer is conditioned on $Y_0$. The two endpoint randomizers remain hidden from the intermediate access hubs.

\noindent\textbf{4) $\mathsf{Solve}$.}
$H_C$ decrypts the $C_2$ received during $\mathsf{Convert}$ to obtain $x_2$. It invokes $\mathsf{ClaimTransfer}$ on $c_{H_S,H_C}$ when $H_S\neq H_C$ and on $c_{S,H_C}$ when $H_S=H_C$. When $H_S\neq H_C$, $H_S$ extracts $x_2$ from the completed transfer signature and invokes $\mathsf{ClaimTransfer}$ on $c_{S,H_S}$. When $H_S=H_C$, the completed signature is delivered directly to $S$. $S$ extracts $x_2$, computes $x_1=\alpha_S^{-1}x_2$, and sends $x_1$ directly to $R$. $R$ computes $x_0=\alpha_R^{-1}x_1$ and invokes $\mathsf{ClaimTransfer}$ on $c_{H_R,R}$. Finally, when $H_R\neq H_C$, $H_R$ extracts $x_0$ from that completed transfer signature and invokes $\mathsf{ClaimTransfer}$ on $c_{H_C,H_R}$. Each claim is installed through $\mathsf{Apply}$ after its completed transfer signature is checked against the stored task. The extracted witnesses then propagate along the settlement path.

\subsection{Failure Handling}
\label{sec:unipch-recovery}

UniPCH normally completes through the off-chain flow above. If cooperative execution stops, recovery follows the latest jointly confirmed state on each affected channel. Before the applicable refund deadline, an affected participant submits this state to the underlying ledger with sufficient time for it to become enforceable. An escrow that has not been converted is then recovered through $\mathsf{RefundEscrow}$ after its deadline, while a transfer whose witness has not appeared is recovered through $\mathsf{RefundTransfer}$ after expiration.

With three distinct hubs, claims in $\mathsf{Solve}$ proceed on the channels between $H_S$ and $H_C$, between $S$ and $H_S$, between $H_R$ and $R$, and finally between $H_C$ and $H_R$. The refund deadlines follow the same order so that each subsequent claimant has sufficient time to obtain the required witness and complete its claim:
\begin{equation}
 \tau_{H_S,H_C}<\tau_{S,H_S}<
 \tau_{H_R,R}<\tau_{H_C,H_R},
 \label{eq:unipch-timeouts}
\end{equation}
Each gap provides sufficient time for state submission, witness extraction or construction, the required message delivery, and confirmation on the corresponding ledger.

Each UniPC realization maps the two refund operations and $\mathsf{ClaimTransfer}$ to terminal actions verifiable by the ledger. The UTXO realization uses its commitment transaction, timelocked item scripts, and mechanisms such as revocation keys. The account realization uses sequence numbers and a challenge period in which a jointly signed state with a higher version supersedes one with a lower version.

\section{Security Analysis}
\label{sec:security}

This section defines local UniPC security, griefing resistance, economic atomicity, and session unlinkability, and proves that UniPCH satisfies these properties.

\begin{definition}[Local UniPC Security]
\label{def:unipch-local-security}
UniPCH has local UniPC security if, across all polynomially many protocol sessions, no honest participant on a channel edge of the corresponding logical path $\mathcal{Q}$ accepts an installed or terminal channel record that violates transition correctness, authorization security, state freshness, or settlement consistency, except with negligible probability in $\lambda$.
\end{definition}

\begin{definition}[Griefing Resistance]
\label{def:unipch-griefing}
UniPCH resists griefing by malicious receiver parties if, except with negligible probability in $\lambda$, every receiver prepayment installed by an honest $H_C$ is backed by a distinct $\mathsf{Reserve}$ session whose sender reservation updates of the agreed amount have been jointly confirmed, and no valid credential authorizes more than one such prepayment.

\end{definition}

\begin{definition}[Economic Atomicity]
\label{def:unipch-atomicity}
UniPCH is economically atomic if every honest endpoint either completes its intended transfer or recovers its locked funds, and whenever an honest intermediate hub's outgoing item is claimed, it is backed by a distinct incoming item of the agreed value that has been claimed or remains claimable before its deadline. An unclaimed outgoing item is refundable.
\end{definition}

\begin{definition}[Session Unlinkability]
\label{def:unipch-unlinkability}
Within one epoch, the selected coordinating hub and the agreed payment amount are public. Let $\mathcal{L}=(\mathsf{epoch},H_C,\mathsf{amount})$ denote this public leakage, and let payments with the same $\mathcal L$ form the anonymity set. The adversary selects two bijections between sets of honest senders and receivers with the same size such that all candidate payments have the same $\mathcal L$. The challenger executes one bijection according to a hidden bit $b$ and gives the adversary the combined protocol transcripts observed by the participating honest-but-curious hubs and the observations from public ledgers, while messages carried by the anonymous endpoint channel remain hidden. UniPCH is unlinkable if the adversary's advantage $|\Pr[b'=b]-1/2|$ is negligible in $\lambda$.
\end{definition}

\begin{theorem}[Local UniPC Security]
\label{thm:unipch-local-security}
If every channel on each payment path $\mathcal{Q}$ uses a UniPC realization satisfying Definition~\ref{def:unipc-security}, then UniPCH satisfies Definition~\ref{def:unipch-local-security} over all polynomially many protocol sessions.
\end{theorem}

\begin{proof}[Proof of Theorem~\ref{thm:unipch-local-security}]
Each channel edge on a payment path satisfies the four properties in Definition~\ref{def:unipc-security} except with negligible probability. A union bound over the polynomially many path edges and sessions preserves negligibility, which gives Definition~\ref{def:unipch-local-security}.
\end{proof}

\begin{theorem}[Griefing Resistance]
\label{thm:unipch-griefing}
Suppose the sender participants and an honest $H_C$ follow the protocol, and $\Pi_{\mathrm{BRS}}$ is unforgeable. Then UniPCH satisfies Definition~\ref{def:unipch-griefing} against malicious receiver parties.
\end{theorem}

\begin{proof}[Proof of Theorem~\ref{thm:unipch-griefing}]
Consider any receiver prepayment installed by an honest $H_C$. Its request contains a valid credential $(\mathsf{tid},\sigma_{\mathsf{tid}})$. By the unforgeability of $\Pi_{\mathrm{BRS}}$, this credential must be derived from a blind signature response produced by $H_C$ in a $\mathsf{Reserve}$ session, except with negligible probability. Because the sender participants follow $\mathsf{Reserve}$, $S$ releases the credential to $R$ only after the corresponding sender reservation updates are jointly confirmed. Before installation, $H_C$ checks whether $\mathsf{tid}$ has already been used and records it. Any subsequent reuse of the credential, including a rerandomized form, is rejected. Since each $\mathsf{Reserve}$ session yields one credential, distinct installed prepayments require distinct confirmed $\mathsf{Reserve}$ sessions. UniPCH therefore satisfies Definition~\ref{def:unipch-griefing}.
\end{proof}

\begin{theorem}[Economic Atomicity]
\label{thm:unipch-atomicity}
Suppose every channel on $\mathcal{Q}$ satisfies Definition~\ref{def:unipc-security}, and $\Pi_{\mathrm{S\text{-}RBCS}}^{*}$ satisfies the correctness, task unforgeability, and unlockability clauses of Definition~\ref{def:srbcs-security}. Suppose also that the four payment phases follow their prescribed order and every successive pair in the claim order satisfies \eqref{eq:security-timeout-gap} under the blockchain assumption of Section~\ref{sec:overview}. Then UniPCH satisfies Definition~\ref{def:unipch-atomicity}.
\end{theorem}

We first establish a lemma on settlement order.

For three distinct hubs, let $e_1=(H_S,H_C)$, $e_2=(S,H_S)$, $e_3=(H_R,R)$, and $e_4=(H_C,H_R)$ denote the claim order in $\mathsf{Solve}$, and let $\tau_i$ be the refund deadline on $e_i$. When roles coincide, the corresponding edges are omitted and the remaining sequence is renumbered as $(e_1,\ldots,e_m)$. Let $\Delta_i$ bound the time required for state submission, witness extraction or construction, message delivery, and ledger confirmation between the claims on $e_i$ and $e_{i+1}$, and let $\epsilon_i>0$ be a scheduling margin. The ordering in \eqref{eq:unipch-timeouts} is instantiated so that
\begin{equation}
 \tau_{i+1}-\tau_i\geq\Delta_i+\epsilon_i,\qquad 1\leq i<m.
 \label{eq:security-timeout-gap}
\end{equation}

The first three phases establish a protection invariant across channels before settlement. $\mathsf{Reserve}$ jointly confirms sender escrows of the agreed amount, $\mathsf{Promise}$ installs receiver transfers without releasing their claim witnesses, and $\mathsf{Convert}$ replaces the sender escrows with transfers linked to independently randomized conditions. A claim on $c_{H_S,H_C}$ releases $x_2$ to $H_S$, enabling the matching claim on $c_{S,H_S}$. A claim on $c_{H_R,R}$ releases $x_0$ to $H_R$, enabling the matching claim on $c_{H_C,H_R}$. The latter claim becomes possible after $x_0$ reaches $H_R$ through the receiver chain, following the settlement initiated by $H_C$ on $c_{H_S,H_C}$. For an honest access hub, the confirmation order establishes the corresponding incoming transfer before the outgoing authorization is issued. For an honest coordinating hub, conditional completion unforgeability bounds completed outgoing authorizations by successful incoming solves. If execution stops earlier, the unclaimed outgoing item remains refundable.

\begin{lemma}[Ordered Settlement Safety]
\label{lem:ordered-settlement}
Under \eqref{eq:security-timeout-gap}, a completed claim on $e_i$ releases the material required for an honest participant to complete the next claim on $e_{i+1}$ before its deadline. If this propagation stops, an honest payer of an unclaimed transfer can enforce the latest state and recover the locked amount through the corresponding refund operation.
\end{lemma}

\begin{proof}[Proof of Lemma~\ref{lem:ordered-settlement}]
A claim by $H_C$ on $c_{H_S,H_C}$ releases $x_2$ to $H_S$, which claims on $c_{S,H_S}$ and releases $x_2$ to $S$. The sender derives $x_1$ and sends it to $R$, which derives $x_0$ and claims on $c_{H_R,R}$, releasing $x_0$ to $H_R$ for the final claim on $c_{H_C,H_R}$. Equation~\eqref{eq:security-timeout-gap} leaves sufficient time for each successive honest action, while $\tau_{H_S,H_C}<\tau_{H_C,H_R}$ prevents $H_C$ from first refunding its outgoing transfer and then claiming its incoming transfer. If propagation stops, the honest payer of the corresponding unclaimed item can recover it through the corresponding refund operation. The same argument applies when edges are omitted because adjacent roles coincide.
\end{proof}

\begin{proof}[Proof of Theorem~\ref{thm:unipch-atomicity}]
Before a candidate has a valid joint confirmation certificate, the preceding jointly confirmed state remains available for recovery. Definition~\ref{def:unipc-security} excludes invalid, incompletely authorized, obsolete, or conflicting local outcomes. Task unforgeability further excludes valid complete signatures on messages that have not been authorized by an honest signer.

For an honest $H_C$, consider any protocol prefix ending with an outgoing authorization completion. If distinct outgoing completions exceeded successful incoming solves, an adversary stopping at this prefix would win the conditional-completion challenge. Task unforgeability excludes this event except with negligible probability. Each completed outgoing authorization can therefore be assigned a distinct preceding successful incoming solve. The corresponding incoming item is jointly confirmed, and its completed signature enables the honest hub to claim it under UniPC's settlement rules and the assumed timing bounds. For an honest $H_S$, $\mathsf{Convert}$ jointly confirms the transfer on $c_{S,H_S}$ before authorizing the corresponding transfer on $c_{H_S,H_C}$. For an honest $H_R$, $\mathsf{Promise}$ confirms the transfer on $c_{H_C,H_R}$ before authorizing the corresponding transfer on $c_{H_R,R}$. The adjacent transfers use the same statement $Y_2$ on the sender side and $Y_0$ on the receiver side. Completing an access hub's outgoing authorization therefore reveals, through unlockability, the witness for its confirmed incoming transfer.

Lemma~\ref{lem:ordered-settlement} provides sufficient time to enforce the latest state and complete each successive honest claim. If cooperative execution stops, an honest payer can recover an unclaimed escrow or transfer through the corresponding refund operation after expiration. UniPCH therefore satisfies Definition~\ref{def:unipch-atomicity}.
\end{proof}

\begin{theorem}[Session Unlinkability]
\label{thm:unipch-unlinkability}
If $\Pi_{\mathrm{S\text{-}RBCS}}^{*}$ satisfies the blindness clause of Definition~\ref{def:srbcs-security}, then UniPCH satisfies Definition~\ref{def:unipch-unlinkability} within the anonymity set determined by $\mathcal L$.
\end{theorem}

\begin{proof}[Proof of Theorem~\ref{thm:unipch-unlinkability}]
Consider two challenge pairings with the same public inputs and phase schedule. The S-RBCS blindness experiment includes credential issuance and presentation, randomized pairs of ciphertexts and corresponding statements, adjacent task evidence, and the corresponding public ledger observations. UniPCH obtains these objects through the procedures of $\Pi_{\mathrm{S\text{-}RBCS}}^{*}$. The remaining channel records and confirmation messages are generated through the same UniPC interfaces from the local inputs and protocol randomness. UniPC channels may differ in which task evidence they retain or refresh, with evidence exchanged between adjacent participants. S-RBCS blindness covers these evidence patterns alongside independently randomized credentials and pairs of ciphertexts and corresponding statements, preventing colluding hubs from linking senders and receivers through refresh patterns. An adversary distinguishing the two UniPCH pairings would also distinguish the corresponding S-RBCS views. For larger anonymity sets, any two bijections can be connected by polynomially many pairwise swaps. A non-negligible distinguishing advantage would give a non-negligible advantage for at least one such swap, contradicting S-RBCS blindness. Therefore, UniPCH satisfies Definition~\ref{def:unipch-unlinkability}.
\end{proof}

\section{Experiments}
\label{sec:experiments}
We evaluate UniPCH in three respects. We compare its running time and bandwidth in each phase with representative PCH schemes, measure the effect of network delay, and analyze how the numbers of blockchains, users per blockchain, and payments per user affect its running time and bandwidth.

\subsection{Experimental Setup}
\label{sec:experimental-setup}

\noindent\textbf{Experimental environment.}
We implement UniPCH in Go with UTXO and account realizations. Adaptor signatures and discrete logarithm relations are instantiated over secp256k1, blind randomizable signatures use the scheme of Pointcheval-Sanders~\cite{PS} over BLS12-381, linear-only encryption uses the CL cryptosystem~\cite{CL}, and commitments use the Pedersen construction~\cite{Pedersen}. Zero-knowledge proofs~\cite{NIZK} use the Fiat-Shamir transformation~\cite{Fiat} on $\Sigma$-protocols~\cite{sigmaprotocol}. The experiments run on a Windows 11 Home machine with a 2.30 GHz 12th Gen Intel Core i7-12700H processor, 14 cores, and 24 GB of memory. All compared schemes use the same experimental environment and, for shared cryptographic primitives, the same implementations and parameters.

\noindent\textbf{Baselines.}
We compare UniPCH with A$^2$L$^+$~\cite{A2LPlus}, P2C2T~\cite{P2C2T}, and XPCH~\cite{XPCH}. A$^2$L$^+$ constructs private payments for a single PCH using anonymous atomic locks. P2C2T supports private cross-chain payments through threshold anonymous atomic locks. XPCH connects multiple PCHs through adaptor payment channels to coordinate private cross-chain payments. The default deployment contains six blockchains, with 50 users per blockchain and five payments initiated by each user. UniPCH uses three Bitcoin instances running in regtest mode. We include three Ethereum instances provided by Ganache to evaluate the account realization. Each baseline uses three Bitcoin and three Litecoin instances running in regtest mode. All schemes use the same payment set. We also evaluate UniPCH using the same Bitcoin and Litecoin deployment as the baselines. The results are provided in Appendix~\ref{sec:experiment-bitcoin-litecoin}. For both UniPCH and XPCH, N and F denote the Near and Far paths, respectively. In UniPCH, B and E indicate that the channels on $B_h$ between hubs use the Bitcoin and Ethereum realizations, respectively.

Bitcoin regtest nodes provide UTXOs and confirm channel funding transactions, while Ganache hosts the account channel contracts and processes channel opening and deposits. In the account realization, channels between users and hubs use ordered list hashes because they typically maintain few active items. Channels between hubs aggregate payments from multiple users and use Merkle roots. The Go prototype executes and measures off-chain payment phases after channel setup.

\subsection{Comparison Experiment}
\label{sec:experiment-local}

Table~\ref{tab:local-phase-comparison} reports the average running time and bandwidth per payment in the four protocol phases. XPCH's $\mathsf{SolveS}$ and $\mathsf{SolveR}$ steps are reported as $\mathsf{Solve}$. A$^2$L$^+$ and P2C2T expose Registration, Promise, and Solving, so their Solving costs are reported across the $\mathsf{Convert}$ and $\mathsf{Solve}$ columns. UniPCH-FR uses the same account realization as E but regenerates and verifies evidence for all active transfer tasks at each update.

\begin{table}[ht]
    \centering
    \caption{Comparison of running time and bandwidth.}
    \label{tab:local-phase-comparison}
    \footnotesize
    \setlength{\tabcolsep}{1.0pt}
    \renewcommand{\arraystretch}{1.05}
    \begin{tabularx}{\columnwidth}{|>{\hsize=1.2\hsize\linewidth=\hsize\centering\arraybackslash}X|>{\hsize=.8\hsize\linewidth=\hsize\centering\arraybackslash}X|>{\centering\arraybackslash}X|>{\centering\arraybackslash}X|>{\centering\arraybackslash}X|>{\centering\arraybackslash}X|>{\centering\arraybackslash}X|}
    \hline
    \multicolumn{2}{|c|}{\textbf{Running Time (ms)}} & $\mathsf{Reserve}$ & $\mathsf{Promise}$ & $\mathsf{Convert}$ & $\mathsf{Solve}$ & Total \\
    \hline
    \multicolumn{2}{|c|}{A$^2$L$^+$} & 2.803 & 428.4 & \multicolumn{2}{c|}{254.6} & 685.9 \\
    \hline
    \multicolumn{2}{|c|}{P2C2T} & 4441.5 & 5398.2 & \multicolumn{2}{c|}{792.7} & 10632 \\
    \hline
    \multirow{2}{*}{XPCH} & N & 2.597 & 449.2 & 134.1 & 147.7 & 733.5 \\
    \cline{2-7}
     & F & 3.048 & 483.6 & 243.7 & 295.9 & 1026.2 \\
    \hline
    \multirow{4}{*}{UniPCH} & BN & 3.432 & 444.7 & 130.1 & 144.4 & 722.7 \\
    \cline{2-7}
     & BF & 5.268 & 485.5 & 244.5 & 294.2 & 1029.4 \\
    \cline{2-7}
     & EN & 3.607 & 458.2 & 131.4 & 147.5 & 740.7 \\
    \cline{2-7}
     & EF & 4.512 & 467.8 & 134.9 & 149.4 & 756.6 \\
     \hline
     \multirow{2}{*}{UniPCH-FR} & EN & 3.693 & 453.5 & 132.0 & 152.1 & 741.3 \\
    \cline{2-7}
     & EF & 4.473 & 498.5 & 246.4 & 294.6 & 1043.9 \\
    \hline
    \end{tabularx}
    \vspace{0.5em}

    \begin{tabularx}{\columnwidth}{|>{\hsize=1.2\hsize\linewidth=\hsize\centering\arraybackslash}X|>{\hsize=.8\hsize\linewidth=\hsize\centering\arraybackslash}X|>{\centering\arraybackslash}X|>{\centering\arraybackslash}X|>{\centering\arraybackslash}X|>{\centering\arraybackslash}X|>{\centering\arraybackslash}X|}
    \hline
    \multicolumn{2}{|c|}{\textbf{Bandwidth (KB)}} & $\mathsf{Reserve}$ & $\mathsf{Promise}$ & $\mathsf{Convert}$ & $\mathsf{Solve}$ & Total \\
    \hline
    \multicolumn{2}{|c|}{A$^2$L$^+$} & 0.720 & 13.06 & \multicolumn{2}{c|}{6.252} & 20.03 \\
    \hline
    \multicolumn{2}{|c|}{P2C2T} & 62.17 & 83.73 & \multicolumn{2}{c|}{15.28} & 161.2 \\
    \hline
    \multirow{2}{*}{XPCH} & N & 0.735 & 13.86 & 7.746 & 2.226 & 24.57 \\
    \cline{2-7}
     & F & 1.346 & 28.15 & 34.48 & 29.61 & 93.59 \\
    \hline
    \multirow{4}{*}{UniPCH} & BN & 1.398 & 14.48 & 7.727 & 3.178 & 26.78 \\
    \cline{2-7}
     & BF & 2.504 & 31.96 & 42.70 & 42.08 & 119.2 \\
    \cline{2-7}
     & EN & 1.398 & 14.48 & 7.727 & 3.178 & 26.78 \\
    \cline{2-7}
     & EF & 2.435 & 22.59 & 14.65 & 4.735 & 44.42 \\
     \hline
     \multirow{2}{*}{UniPCH-FR} & EN & 1.398 & 14.71 & 8.602 & 4.119 & 28.83 \\
    \cline{2-7}
     & EF & 2.435 & 32.04 & 43.43 & 42.99 & 120.9 \\
     \hline
    \end{tabularx}
\end{table}

\noindent\textbf{Running time.}
A$^2$L$^+$ has the lowest running time. UniPCH supports multiple payments on the same channel within an epoch and coordinates payments across PCHs. P2C2T incurs two-party signing, auxiliary encryption, and VTD operations. All tested UniPCH configurations are faster than P2C2T. BN and EN take 1.5\% less and 1.0\% more time than XPCH-N, respectively. BF increases it over XPCH-F by 0.3\%. UTXO updates refresh evidence for all active transfer tasks. The account realization retains unchanged evidence, allowing EF to reduce the running time by 26.3\% relative to XPCH-F and 26.5\% relative to BF. Under the same account realization, UniPCH and FR have similar running times on the Near path. Selective refresh reduces the running time by 27.5\% on the Far path compared with FR.

\noindent\textbf{Bandwidth.}
A$^2$L$^+$ uses the least bandwidth at 20.03 KB. P2C2T requires 161.2 KB, including auxiliary ciphertexts and VTD material. The common UniPC interfaces carry task evidence and additional protocol fields. BN and EN use 9.0\% more bandwidth than XPCH-N. Near operations do not use channels between hubs, so changing $B_h$ leaves their user channel realizations and message flow unchanged. UniPCH-BN and UniPCH-EN report the same bandwidth in each phase. BF uses 27.4\% more than XPCH-F. EF uses 52.5\% less bandwidth than XPCH-F and 62.8\% less than BF. Under the same account realization, selective refresh reduces bandwidth by 7.1\% on the Near path and 63.3\% on the Far path compared with FR. UniPCH-FR sends refreshed evidence for existing transfers on account user channels in Promise, Convert, and Solve, increasing Near bandwidth over UniPCH-BN and UniPCH-EN. Reserve precedes transfer creation, so its bandwidth matches that of the other two configurations.

\begin{table}[ht]
    \centering
    \caption{Average adaptor pre-signature generations per payment.}
    \label{tab:presign-count}
    \footnotesize
    \renewcommand{\arraystretch}{1.3}
    \setlength{\tabcolsep}{12pt}
    \begin{tabular}{|c|c|c|c|}
        \hline
        Scheme & UniPCH-B & UniPCH-E & UniPCH-FR-E \\
        \hline
        Count & 294.694 & 14.079 & 304.277 \\
        \hline
    \end{tabular}
\end{table}

\noindent\textbf{Adaptor pre-signature generation.}
Table~\ref{tab:presign-count} reports the average number of adaptor pre-signatures generated per payment. UniPCH-B and UniPCH-E use the same endpoint channel realizations. By reusing evidence on the channels between hubs, UniPCH-E generates 95.2\% fewer pre-signatures than UniPCH-B. With the same account realization, selective refresh reduces the count by 95.4\% relative to UniPCH-FR-E. Furthermore, on account channels between users and hubs, UniPCH-FR-E refreshes evidence for all active transfers, while UniPCH-B retains unchanged evidence. UniPCH-FR-E therefore generates slightly more pre-signatures than UniPCH-B.

\begin{table}[ht]
    \centering
    \caption{Running time under emulated network delay.}
    \label{tab:delayed-phase-comparison}
    \footnotesize
    \setlength{\tabcolsep}{1.0pt}
    \renewcommand{\arraystretch}{1.05}
    \begin{tabularx}{\columnwidth}{|>{\hsize=1.2\hsize\linewidth=\hsize\centering\arraybackslash}X|>{\hsize=.8\hsize\linewidth=\hsize\centering\arraybackslash}X|>{\centering\arraybackslash}X|>{\centering\arraybackslash}X|>{\centering\arraybackslash}X|>{\centering\arraybackslash}X|>{\centering\arraybackslash}X|}
    \hline
    \multicolumn{2}{|c|}{\textbf{Running Time (s)}} & $\mathsf{Reserve}$ & $\mathsf{Promise}$ & $\mathsf{Convert}$ & $\mathsf{Solve}$ & Total \\
    \hline
    \multirow{2}{*}{XPCH} & N & 2.027 & 2.863 & 1.744 & 3.252 & 9.886 \\
    \cline{2-7}
     & F & 3.245 & 4.693 & 3.834 & 6.126 & 17.899 \\
    \hline
    \multirow{4}{*}{UniPCH} & BN & 1.625 & 2.491 & 1.341 & 2.976 & 8.433 \\
    \cline{2-7}
     & BF & 2.843 & 4.430 & 3.058 & 6.059 & 16.389 \\
    \cline{2-7}
     & EN & 1.628 & 2.505 & 1.344 & 2.978 & 8.455 \\
    \cline{2-7}
     & EF & 2.846 & 4.131 & 2.565 & 5.414 & 14.955 \\
    \hline
    \end{tabularx}
\end{table}

\subsection{Evaluation of Network Delay}
\label{sec:experiment-delay}

We emulate network delay in WSL using Linux traffic control (tc), with a 200 ms delay and 250 Mbit/s bandwidth. Table~\ref{tab:delayed-phase-comparison} reports the running times of XPCH and UniPCH. BN and EN reduce the running time relative to XPCH-N by 14.7\% and 14.5\%, respectively. BF and EF reduce it relative to XPCH-F by 8.4\% and 16.4\%, respectively. UniPC's common update flow reduces message exchanges across the protocol phases. BF's improvement is smaller than EF's because UTXO updates require fresh evidence for all active transfer tasks.

\subsection{Parameter Sensitivity}
\label{sec:experiment-parameters}

\begin{figure}[ht]
    \centering
    \includegraphics[width=\columnwidth]{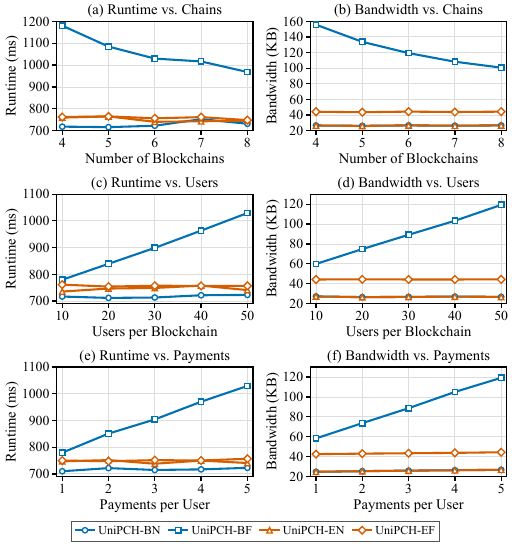}
    \caption{Effect of different parameters on running time and bandwidth.}
    \label{fig:exp-param}
\end{figure}

Fig.~\ref{fig:exp-param} varies one parameter while retaining the default values of the other two. The blockchain experiment starts with two Bitcoin regtest instances and two Ethereum Ganache instances, then alternately adds an Ethereum and a Bitcoin instance until reaching eight blockchains. BN, EN, and EF vary little as the number grows from four to eight. BF decreases from 1.180 s to 0.968 s and from 155.6 KB to 100.7 KB because the evidence refresh workload is distributed over more UTXO channels between hubs. Increasing the users per blockchain from 10 to 50 or the payments per user from one to five raises BF to 1.029 s and 119.2 KB, while the Near and EF curves change relatively little. Both changes leave more active transfers on shared UTXO channels. Each update then requires more task evidence generation and verification, increasing running time and bandwidth. Stable item identifiers allow the account realization to retain unchanged evidence.

In addition, we evaluate the on-chain costs of channel operations and dispute settlement for both UniPC realizations in Appendix~\ref{sec:onchain-cost}.

\section{Discussion}

\noindent\textbf{Cost of Different Realizations.}
Let $k$ denote the number of active transfers in a candidate channel state. Replacing a UTXO commitment changes the signing messages of all active transfers, requiring $O(k)$ adaptor pre-signature generations per update. The account realization retains unchanged evidence and generates only $O(1)$ adaptor pre-signatures per operation. Because channels between hubs aggregate payments from multiple users, active transfers accumulate within an epoch and increase the UTXO refresh workload.

\noindent\textbf{Interaction Cost.}
UniPC organizes candidate state construction, task evidence, and state confirmation in a common update flow. Compared with XPCH's process, which is organized around transactions, this flow requires fewer message exchanges to authorize a state change and is therefore less sensitive to network delay. UTXO evidence refresh adds computation, so BF achieves a smaller reduction in running time than EF under network delay.

\noindent\textbf{Ledger Delays and Timeouts.}
Differences in block production and confirmation times affect the spacing of refund deadlines. For each pair of successive claims, the delay bound $\Delta_i$ in \eqref{eq:security-timeout-gap} includes the relevant ledger confirmation bounds and any waiting required to make the channel state enforceable, together with witness processing and message delivery. Slower confirmation or longer dispute periods require wider gaps between refund deadlines and may prolong fund locking during recovery, while the claim order remains unchanged.

\noindent\textbf{Integration of New Blockchains.}
Consider a new blockchain $B'$ with a designated hub $H'$ and users $\mathcal{U}'$. Through a conforming UniPC realization, each user establishes and funds a channel with $H'$ on $B'$. The hub generates and publishes its LOE and BRS public keys and establishes funded UniPC channels with all existing hubs on $B_h$. Once these steps are completed between epochs, $H'$ and its users can participate from the next epoch.

\noindent\textbf{Scalability Considerations.}
With $n$ hubs, the fully connected overlay maintains $n(n-1)/2$ channels on $B_h$, while each payment uses at most four channels. Expanding the hub set increases channel maintenance and funding requirements. Selective refresh reduces authorization work on shared channels, but these channels process state updates sequentially and reserve liquidity for pending payments.

\section{Conclusion}
We present UniPCH, a framework with two layers that separates native channel execution from private cross-chain payment logic. This separation enables payments between users on heterogeneous blockchains through channels with different native update and settlement mechanisms. UniPC defines common channel operations and security requirements, with conforming UTXO and account realizations. The common interfaces allow these realizations to reuse the upper payment flow and its security analysis. UniPC retains valid task evidence across updates. We prove local UniPC security, griefing resistance, economic atomicity, and session unlinkability within an epoch. Experiments show modest computational overhead from the common interfaces and reduced evidence generation and exchange through selective refresh.

\section*{Acknowledgments}
This work was supported in part by the National Natural Science Foundation of China (NSFC) / Research Grants Council (RGC) Collaborative Research Scheme (Grant No. 62461160332 \& CRS\_HKUST602/24), and in part by the Major Key Project of Pengcheng Laboratory under Grant PCL2025A07.

\bibliographystyle{IEEEtran}
\bibliography{myref}

\appendices

\section{Complete Security Definitions}
\label{sec:complete-security-definitions}

\subsection{UniPC Security}
\label{sec:complete-unipc-security}

Consider a PPT adversary $\mathcal{A}$ against a conforming UniPC realization $\mathcal{R}$. The experiment $\mathsf{Exp}^{\mathrm{UniPC}}_{\mathcal{R},\mathcal{A}}(1^\lambda)$ initializes a channel with at least one honest participant. The adversary controls any corrupted participant and may schedule off-chain messages arbitrarily, but cannot forge the origin of messages sent over authenticated channels. The underlying ledger correctly validates and executes transactions or contract calls, enforces the applicable scripts, timelocks, or contract state transitions prescribed by the realization, and confirms valid submissions within the assumed delay bound. This allows an honest participant to exercise the prescribed stale-state response within the applicable response or challenge period. The adversary may initiate polynomially many UniPC operations and submit candidate states, task evidence, and terminal actions. Each honest participant follows $\mathcal{R}$, including its prescribed response to an obsolete state, as well as UniPC's update ordering and confirmation rules. Let $\mathcal{T}$ be the resulting trace of installed records and ledger-accepted terminal actions. We define the following events over $\mathcal{T}$.

\begin{itemize}
\item $\mathsf{Bad}_{\mathrm{tr}}$ occurs if an installed state is not obtained from the preceding installed state by an admissible operation in \eqref{eq:unipc-transitions}, or violates the conservation invariant in \eqref{eq:unipc-conservation}.
\item $\mathsf{Bad}_{\mathrm{auth}}$ occurs if a record $(\Sigma',\mathcal{P}',\mathcal{E}',\Gamma')$ is installed while $\mathsf{Comp}(\mathcal{P}',\mathcal{E}')=0$ or $\mathsf{VfyConf}(\Sigma',\Gamma')=0$. It also occurs if evidence retained through $\mathcal{P}_{\mathrm{keep}}$ is accepted for a different task $\eta'\neq\eta$, where tasks are compared as complete tuples $(\mathsf{signer},m,Y)$.
\item $\mathsf{Bad}_{\mathrm{old}}$ occurs if an obsolete record becomes the terminal ledger state although an honest participant invokes the stale-state response prescribed by $\mathcal{R}$ within the applicable response bound.
\item $\mathsf{Bad}_{\mathrm{term}}$ occurs if the ledger accepts conflicting claim and refund outcomes for the same item in an installed record.
\end{itemize}

For $z\in\{\mathrm{tr},\mathrm{auth},\mathrm{old},\mathrm{term}\}$, define
\begin{equation}
 \mathsf{Adv}^{z}_{\mathcal{R},\mathcal{A}}(\lambda)
 =\Pr[\mathsf{Bad}_{z}],
 \label{eq:formal-unipc-advantage}
\end{equation}
where the probability is over the randomness of the experiment.

\begin{definition}[Formal UniPC Security]
\label{def:formal-unipc-security}
A conforming realization $\mathcal{R}$ satisfies transition correctness, authorization security, state freshness, and settlement consistency if $\mathsf{Adv}^{\mathrm{tr}}_{\mathcal{R},\mathcal{A}}$, $\mathsf{Adv}^{\mathrm{auth}}_{\mathcal{R},\mathcal{A}}$, $\mathsf{Adv}^{\mathrm{old}}_{\mathcal{R},\mathcal{A}}$, and $\mathsf{Adv}^{\mathrm{term}}_{\mathcal{R},\mathcal{A}}$ are negligible in $\lambda$, respectively, for every PPT adversary $\mathcal{A}$. It is secure if all four advantages are negligible.
\end{definition}

\subsection{S-RBCS Security}
\label{sec:complete-srbcs-security}

Let $\Pi_{\mathrm{S\text{-}RBCS}}=(\mathsf{Setup},\mathsf{Reserve},\mathsf{Promise},\mathsf{Convert},\mathsf{Solve})$ be an S-RBCS scheme with the syntax defined in Section~\ref{sec:unipch-crypto}. The following definitions formalize the four requirements in Definition~\ref{def:srbcs-security}.

\begin{definition}[S-RBCS Correctness]
\label{def:srbcs-correctness}
$\Pi_{\mathrm{S\text{-}RBCS}}$ is correct if an honest execution satisfies the following conditions except with negligible probability: every credential produced by $\mathsf{Reserve}$ verifies after unblinding and randomization; each ciphertext-statement pair produced in $\mathsf{Promise}$ or $\mathsf{Convert}$ encodes a common witness within that pair; retained and refreshed evidence for every required task $\eta$ satisfies $\mathsf{VfyEv}(\eta,e)=1$; and adapting valid evidence yields a valid signature from which $\mathsf{Extract}$ recovers a witness for its statement.
\end{definition}

For task unforgeability, $\mathcal{A}$ chooses one of the following two challenges. All interactions follow the prescribed epoch and phase structure, and the honest parties execute the procedures of $\Pi_{\mathrm{S\text{-}RBCS}}$.

\emph{Fresh-message challenge.} The challenger initializes a designated honest signer and gives its public key to $\mathcal{A}$. The adversary may participate in polynomially many protocol executions and obtain the resulting signatures and task evidence. The challenger records in $\mathcal{Q}_{\mathrm{msg}}$ every message signed, pre-signed, or re-signed under the challenge key. The adversary outputs $(m^*,\sigma^*)$ and wins if $\mathsf{Verify}(pk,m^*,\sigma^*)=1$ and $m^*\notin\mathcal{Q}_{\mathrm{msg}}$.

\emph{Conditional-completion challenge.} The challenger runs the coordinating hub $H_C$ honestly, while $\mathcal{A}$ controls the other participants in Near or Far interactions. The adversary may initiate polynomially many executions of $\mathsf{Reserve}$, $\mathsf{Promise}$, $\mathsf{Convert}$, and $\mathsf{Solve}$, including the evidence updates prescribed by the construction. For each outgoing conditional authorization issued by $H_C$, the challenger creates a record $\ell$ containing its signing key, condition statement, and set $M_\ell$ of authorized messages. When the item authorization, signer, and condition remain unchanged, retaining evidence or refreshing the message updates the existing record. A new conditional authorization creates a new record.

Let $D$ be the number of distinct incoming conditional authorizations for which the coordinating hub's $\mathsf{Solve}$ procedure produces a valid completed signature. An authorization is counted when the signature is produced, even if $\mathcal{A}$ later aborts. The adversary outputs $(\ell_i,m_i,\sigma_i)_{i=1}^{d}$. Refreshed versions of the same authorization count as one output. The adversary wins if the records $\ell_i$ are pairwise distinct, each $m_i\in M_{\ell_i}$, each $\sigma_i$ verifies under the corresponding signing key, and $d>D$.

\begin{definition}[Task Unforgeability]
\label{def:srbcs-task-unforgeability}
Let $\mathsf{Unforge}^{\mathcal{A}}_{\Pi_{\mathrm{S\text{-}RBCS}}}(1^\lambda)$ denote the experiment above, which returns $1$ if $\mathcal{A}$ wins the selected challenge. $\Pi_{\mathrm{S\text{-}RBCS}}$ is task unforgeable if, for every PPT adversary $\mathcal{A}$,
\begin{equation}
 \Pr[\mathsf{Unforge}^{\mathcal{A}}_{\Pi_{\mathrm{S\text{-}RBCS}}}(1^\lambda)=1]
 \leq \mathsf{negl}(\lambda).
 \label{eq:srbcs-unforgeability}
\end{equation}
\end{definition}

In the unlockability experiment, the challenger initializes an honest signer. The adversary then selects a challenge task $\eta^*=(\mathsf{signer}^*,m^*,Y^*)$, for which the challenger generates adaptor evidence $\widetilde{\sigma}^*$ according to the construction. The adversary may obtain signatures and task evidence for other messages, but cannot make another signing or pre-signing query for $m^*$ under that signer's key. It outputs a signature $\sigma^*$ and wins if $\sigma^*$ verifies for $m^*$ but $\mathsf{Extract}(\widetilde{\sigma}^*,\sigma^*,Y^*)$ returns $\bot$ or a value inconsistent with $Y^*$. For linked tasks generated by the construction, each corresponding honest endpoint retains its nonzero local randomizer. The adversary also wins if any such endpoint cannot derive the adjacent witness according to \eqref{eq:unipch-randomization}.

\begin{definition}[Unlockability]
\label{def:srbcs-unlockability}
Let $\mathsf{Unlock}^{\mathcal{A}}_{\Pi_{\mathrm{S\text{-}RBCS}}}(1^\lambda)$ denote the experiment above. $\Pi_{\mathrm{S\text{-}RBCS}}$ is unlockable if the winning probability of every PPT adversary is negligible in $\lambda$.
\end{definition}

For blindness, $\mathcal{A}$ selects honest senders $S_0,S_1$ and receivers $R_0,R_1$ whose challenge sessions have identical public inputs and follow the same phase schedule. The challenger samples $b\xleftarrow{\$}\{0,1\}$ and executes two sessions using $\mathcal{M}_0=\{(S_0,R_0),(S_1,R_1)\}$ when $b=0$, and $\mathcal{M}_1=\{(S_0,R_1),(S_1,R_0)\}$ when $b=1$. The adversary receives the combined S-RBCS transcripts visible to the participating honest-but-curious hubs and the corresponding public-ledger observations. This view includes credential issuance and presentation, randomized ciphertext-statement pairs, and task evidence visible on the hubs' adjacent channels. Messages carried by the anonymous channel between endpoints are hidden. The adversary outputs a bit $b'$.

\begin{definition}[Blindness]
\label{def:srbcs-blindness}
$\Pi_{\mathrm{S\text{-}RBCS}}$ is blind if, for every PPT adversary $\mathcal{A}$,
\begin{equation}
 \mathsf{Adv}^{\mathrm{blind}}_{\Pi_{\mathrm{S\text{-}RBCS}},\mathcal{A}}(\lambda)
 =\left|\Pr[b'=b]-\frac{1}{2}\right|
 \leq \mathsf{negl}(\lambda).
 \label{eq:srbcs-blindness}
\end{equation}
\end{definition}

\section{Proofs of UniPC Realization Security}
\label{sec:security-unipc}

We prove the four properties in Definition~\ref{def:formal-unipc-security} separately for each realization.

\subsection{UTXO Realization}
\label{sec:proof-utxo-realization}

\begin{lemma}[UTXO Transition Correctness]
\label{lem:utxo-transition}
Under the assumptions of Theorem~\ref{thm:utxo-realization}, $\mathsf{Adv}^{\mathrm{tr}}_{\mathcal{R}_{\mathrm{UTXO}},\mathcal{A}}(\lambda)$ is negligible for every PPT adversary $\mathcal{A}$.
\end{lemma}

\begin{proof}
An installed candidate contains both participants' signatures on the same commitment transaction. Before signing, an honest participant reconstructs the commitment and its descendants, applies the requested operation, and checks their amounts, scripts, indices, deadlines, and the invariant in \eqref{eq:unipc-conservation}. Hence, $\mathsf{Bad}_{\mathrm{tr}}$ requires either an honest participant to sign a candidate rejected by these checks or the adversary to forge its signature. The former contradicts the prescribed algorithm, and the latter has negligible probability by the unforgeability of $\Pi_{\mathrm{Sig}}$.
\end{proof}

\begin{lemma}[UTXO Authorization Security]
\label{lem:utxo-authorization}
Under the assumptions of Theorem~\ref{thm:utxo-realization}, $\mathsf{Adv}^{\mathrm{auth}}_{\mathcal{R}_{\mathrm{UTXO}},\mathcal{A}}(\lambda)$ is negligible for every PPT adversary $\mathcal{A}$.
\end{lemma}

\begin{proof}
Replacing $\mathsf{TX}_C$ changes the outpoint, and possibly the index, used by every carried claim transaction. The corresponding authorization messages therefore change, giving $\mathcal{P}_{\mathrm{keep}}=\emptyset$ and $\mathcal{P}_{\mathrm{refresh}}=\mathcal{P}'_{\mathrm{UTXO}}$. An honest participant verifies every refreshed adaptor pre-signature before signing the candidate, and $\mathsf{VfyConf}$ checks both confirmation signatures before installation. A required task without evidence accepted by $\mathsf{VfyEv}$ makes $\mathsf{Comp}=0$, while a missing or invalid confirmation fails $\mathsf{VfyConf}$. Producing a valid confirmation on a state message not signed by the honest participant would violate the unforgeability of $\Pi_{\mathrm{Sig}}$. Since $\mathcal{P}_{\mathrm{keep}}=\emptyset$, no evidence is retained for a different task. Thus, $\mathsf{Bad}_{\mathrm{auth}}$ has negligible probability.
\end{proof}

\begin{lemma}[UTXO State Freshness]
\label{lem:utxo-freshness}
Under the assumptions of Theorem~\ref{thm:utxo-realization}, $\mathsf{Adv}^{\mathrm{old}}_{\mathcal{R}_{\mathrm{UTXO}},\mathcal{A}}(\lambda)$ is negligible for every PPT adversary $\mathcal{A}$.
\end{lemma}

\begin{proof}
When a new commitment replaces the current one, the parties exchange the revocation material for the old commitment. If a party later broadcasts that obsolete commitment, the counterparty can use this material to spend its revocable output before the broadcaster's timelocked settlement branch becomes available. The blockchain confirms the revocation spend within the assumed delay bound. Allowing the obsolete commitment to determine the terminal settlement despite this response would violate transaction validation or the confirmation bound. Thus, $\mathsf{Bad}_{\mathrm{old}}$ occurs only with negligible probability.
\end{proof}

\begin{lemma}[UTXO Settlement Consistency]
\label{lem:utxo-settlement}
Under the assumptions of Theorem~\ref{thm:utxo-realization}, $\mathsf{Adv}^{\mathrm{term}}_{\mathcal{R}_{\mathrm{UTXO}},\mathcal{A}}(\lambda)$ is negligible for every PPT adversary $\mathcal{A}$.
\end{lemma}

\begin{proof}
The claim and refund transactions for an item spend the same output of the installed commitment through alternative script branches. Once one transaction is confirmed, accepting the other would spend the same outpoint twice. This contradicts ledger validation, so conflicting terminal outcomes occur only with negligible probability.
\end{proof}

\begin{proof}[Proof of Theorem~\ref{thm:utxo-realization}]
Lemmas~\ref{lem:utxo-transition}--\ref{lem:utxo-settlement} make all four advantages in Definition~\ref{def:formal-unipc-security} negligible. The UTXO realization therefore satisfies Definition~\ref{def:unipc-security}.
\end{proof}

\subsection{Account Realization}
\label{sec:proof-account-realization}

\begin{lemma}[Account Transition Correctness]
\label{lem:account-transition}
Under the assumptions of Theorem~\ref{thm:account-realization}, $\mathsf{Adv}^{\mathrm{tr}}_{\mathcal{R}_{\mathrm{Account}},\mathcal{A}}(\lambda)$ is negligible for every PPT adversary $\mathcal{A}$.
\end{lemma}

\begin{proof}
The digest $h_{\Sigma}$ binds the channel identifier, version, balances, and commitment to the complete active item set. An honest participant reconstructs the requested operation, checks the candidate and value conservation, and signs only the resulting digest. Installing a state involved in $\mathsf{Bad}_{\mathrm{tr}}$ therefore requires a forgery of its confirmation signature or a collision that binds the same digest to different state contents. Both events have negligible probability.
\end{proof}

\begin{lemma}[Account Authorization Security]
\label{lem:account-authorization}
Under the assumptions of Theorem~\ref{thm:account-realization}, $\mathsf{Adv}^{\mathrm{auth}}_{\mathcal{R}_{\mathrm{Account}},\mathcal{A}}(\lambda)$ is negligible for every PPT adversary $\mathcal{A}$.
\end{lemma}

\begin{proof}
The verifier checks evidence against the complete task $\eta=(\mathsf{signer},m,Y)$. Under the realization's task derivation rules, an unchanged transfer keeps its task tuple and evidence, a new transfer receives a new task and the corresponding evidence, and removing a transfer removes its task. Retained evidence therefore remains assigned to its original task. A required task without evidence accepted by $\mathsf{VfyEv}$ makes $\mathsf{Comp}=0$. After this completeness check, the participants sign the candidate state digest. $\mathsf{VfyConf}$ rejects installation if either confirmation is missing or invalid. Producing a valid confirmation on a state message not signed by the honest participant would violate the unforgeability of $\Pi_{\mathrm{Sig}}$. Thus, $\mathsf{Bad}_{\mathrm{auth}}$ has negligible probability.
\end{proof}

\begin{lemma}[Account State Freshness]
\label{lem:account-freshness}
Under the assumptions of Theorem~\ref{thm:account-realization}, $\mathsf{Adv}^{\mathrm{old}}_{\mathcal{R}_{\mathrm{Account}},\mathcal{A}}(\lambda)$ is negligible for every PPT adversary $\mathcal{A}$.
\end{lemma}

\begin{proof}
During the challenge period, the contract accepts a jointly signed state with a higher version and supersedes every lower submitted version. If an honest participant submits the latest jointly confirmed state within this period, finalizing an older record contradicts the contract's version rule or correct execution by the blockchain. Hence, $\mathsf{Bad}_{\mathrm{old}}$ has negligible probability.
\end{proof}

\begin{lemma}[Account Settlement Consistency]
\label{lem:account-settlement}
Under the assumptions of Theorem~\ref{thm:account-realization}, $\mathsf{Adv}^{\mathrm{term}}_{\mathcal{R}_{\mathrm{Account}},\mathcal{A}}(\lambda)$ is negligible for every PPT adversary $\mathcal{A}$.
\end{lemma}

\begin{proof}
Both claim and refund actions verify item membership in the finalized state and reject a consumed item. A successful action marks the item as consumed before transferring funds, so a later claim or refund for that item is rejected. Accepting both outcomes would contradict the contract state machine or correct ledger execution and therefore occurs only with negligible probability.
\end{proof}

\begin{proof}[Proof of Theorem~\ref{thm:account-realization}]
Lemmas~\ref{lem:account-transition}--\ref{lem:account-settlement} make all four advantages in Definition~\ref{def:formal-unipc-security} negligible. The account realization therefore satisfies Definition~\ref{def:unipc-security}.
\end{proof}

\section{Proof of S-RBCS Security}
\label{sec:security-srbcs}

We prove separately that the construction $\Pi_{\mathrm{S\text{-}RBCS}}^{*}$ satisfies the four requirements in Definitions~\ref{def:srbcs-correctness}-\ref{def:srbcs-blindness}. The arguments use only the task and evidence interfaces of UniPC and therefore do not depend on a particular ledger realization.

\begin{lemma}[S-RBCS Correctness]
\label{lem:srbcs-correctness}
Under the assumptions of Theorem~\ref{thm:srbcs-construction}, $\Pi_{\mathrm{S\text{-}RBCS}}^{*}$ satisfies Definition~\ref{def:srbcs-correctness}.
\end{lemma}

\begin{proof}
The correctness of $\Pi_{\mathrm{BRS}}$ ensures that unblinding and randomizing the response in \eqref{eq:unipch-credential} produces a credential accepted under $H_C$'s public key. The correctness of $\Pi_{\mathrm{LOE}}$ ensures that every honestly generated ciphertext-statement pair contains the same witness, and the completeness of $\Pi_{\mathrm{ZK}}$ ensures that its accompanying proof verifies. Multiplication by the nonzero scalars $\alpha_R$ and $\alpha_S$ preserves this relation in \eqref{eq:unipch-randomization}.

For every task in $\mathcal{P}_{\mathrm{keep}}$, the retained evidence has already been accepted for the unchanged tuple $(\mathsf{signer},m,Y)$. Each task in $\mathcal{P}_{\mathrm{refresh}}$ receives evidence generated by $\mathsf{PreSign}$ or $\mathsf{RSign}$ and is checked by $\mathsf{VfyEv}$. The correctness of $\Pi_{\mathrm{AS}}$ and $\Pi_{\mathrm{RAS}}$ therefore makes all required evidence valid. Finally, adapting valid evidence with $x$ yields a valid signature, and extraction from the evidence and completed signature recovers $x$ such that $Y=xG$. All conditions in Definition~\ref{def:srbcs-correctness} follow.
\end{proof}

\begin{lemma}[S-RBCS Task Unforgeability]
\label{lem:srbcs-task-unforgeability}
Under the assumptions of Theorem~\ref{thm:srbcs-construction}, the winning probability in \eqref{eq:srbcs-unforgeability} is negligible.
\end{lemma}

\begin{proof}
We treat the two challenges in Definition~\ref{def:srbcs-task-unforgeability} separately.

\emph{Fresh message challenge.} A reduction uses an AS or RAS challenge key for the designated honest signer. It answers the permitted signing and task evidence requests with the appropriate oracles and records every authorized message in $\mathcal{Q}_{\mathrm{msg}}$. A refresh is handled as a pre-signing query for the replacement message under the same statement. If $\mathcal{A}$ returns a valid $(m^*,\sigma^*)$ with $m^*\notin\mathcal{Q}_{\mathrm{msg}}$, the same pair is a forgery. The reduction loses only a polynomial factor in guessing the challenge key and evidence type. A non-negligible success probability would therefore contradict the unforgeability of $\Pi_{\mathrm{AS}}$ or $\Pi_{\mathrm{RAS}}$.

\emph{Conditional completion challenge.} Let $N$ be a polynomial upper bound on the number of independent outgoing authorizations created during the experiment. The reduction receives an OMDL challenge containing $N$ independently sampled points $Y_j=r_jG$ whose discrete logarithms it does not know. It generates the coordinating hub's signing and credential keys, assigns a different challenge point to each new authorization record $\ell_j$, and represents the record's initial ciphertext with a fresh LOE handle encrypting the formal variable $r_j$. The accompanying honest proofs are simulated. By zero knowledge, this changes the adversary's success probability by at most a negligible amount.

The reduction records each issued pre-signature with its authorization record, message, and statement. It returns retained evidence unchanged. When the message of an existing authorization changes, the reduction uses the corresponding signing key and the unchanged public statement to generate refreshed evidence. Neither operation queries the discrete logarithm oracle or creates another independent authorization record.

In the LOE model, the reduction records each ciphertext handle as a linear form $a_0+\sum_{j=1}^{N}a_jr_j$. The corresponding public point $a_0G+\sum_{j=1}^{N}a_jY_j$ can be computed from the challenge points without knowing the $r_j$. Encryption of known values, linear operations, and ciphertext randomization are simulated by updating these forms and returning handles with the prescribed distribution.

For each incoming $\mathsf{Solve}$ request, the reduction verifies the task evidence and checks that the task statement matches the public point represented by the ciphertext. It rejects invalid evidence or an inconsistent pair. If the witness for a valid request is not yet known, the reduction queries its discrete logarithm oracle for that point, completes the incoming pre-signature with the returned witness, and stores it for later requests involving the same condition. Retrying or refreshing the same incoming authorization makes no further discrete logarithm query. The authorization is counted in $D$ once its valid completed signature is produced, even if $\mathcal{A}$ later aborts. Hence, the simulation makes at most $D$ discrete logarithm queries.

For each valid output $(\ell_i,m_i,\sigma_i)$, the reduction retrieves the evidence recorded for that authorization and message. Except with negligible probability, witness extractability yields a value $r_i$ satisfying $Y_{\ell_i}=r_iG$. If extraction fails, selecting the affected evidence instance from the polynomially many prescribed queries and forwarding its completed signature gives an adversary against the corresponding AS or RAS extraction experiment. The theorem's message-binding requirement prevents an authorized message from being counted for two independent item authorizations, except with negligible probability.

Suppose that $\mathcal{A}$ wins with $d>D$. Its outputs reveal the discrete logarithms of $d$ distinct challenge points, and the reduction queries its oracle for the remaining $N-d$. It thus recovers all $N$ logarithms with at most $D+(N-d)\leq N-1$ queries. A run with $D\geq N$ cannot satisfy $d>D$, so the reduction may abort before exceeding its query budget. Any non-negligible success probability would contradict OMDL hardness.

Both challenges therefore have negligible winning probability.
\end{proof}

\begin{lemma}[S-RBCS Unlockability]
\label{lem:srbcs-unlockability}
Under the assumptions of Theorem~\ref{thm:srbcs-construction}, $\Pi_{\mathrm{S\text{-}RBCS}}^{*}$ satisfies Definition~\ref{def:srbcs-unlockability}.
\end{lemma}

\begin{proof}
Suppose an adversary wins the unlockability experiment with non-negligible probability. A reduction submits the selected challenge task to the witness-extractability challenger for $\Pi_{\mathrm{AS}}$ or $\Pi_{\mathrm{RAS}}$ and gives $\mathcal{A}$ the resulting challenge pre-signature as $\widetilde{\sigma}^*$. It answers the permitted queries with the appropriate oracles and submits $\mathcal{A}$'s final signature $\sigma^*$ to the challenger. The query restriction ensures that no additional signing or pre-signing query for $m^*$ is made under the challenge key. If $\sigma^*$ is valid but the challenge witness cannot be extracted, the reduction wins the underlying witness-extractability experiment. Guessing the challenge signer and evidence type costs only a polynomial factor, so non-negligible success contradicts the assumed extractability.

For linked tasks, extraction first reveals the witness attached to the completed statement. Equation~\eqref{eq:unipch-randomization} gives $x_2=\alpha_Sx_1$ and $x_1=\alpha_Rx_0$. Since both randomizers are nonzero, an honest sender derives $x_1=\alpha_S^{-1}x_2$, and an honest receiver derives $x_0=\alpha_R^{-1}x_1$. These values satisfy their adjacent statements. Failure of this derivation would contradict field inversion, the correctness of $\Pi_{\mathrm{LOE}}$, or the soundness of $\Pi_{\mathrm{ZK}}$. Thus the winning probability in the unlockability experiment is negligible.
\end{proof}

\begin{lemma}[S-RBCS Blindness]
\label{lem:srbcs-blindness}
Under the assumptions of Theorem~\ref{thm:srbcs-construction}, the advantage in \eqref{eq:srbcs-blindness} is negligible.
\end{lemma}

\begin{proof}
Consider the two endpoint matchings $\mathcal{M}_0$ and $\mathcal{M}_1$ in the blindness experiment. We use the following hybrids.

\emph{Hybrid $\mathcal{H}_0$.} This is the real transcript generated with $\mathcal{M}_0$.

\emph{Hybrid $\mathcal{H}_1$.} In the two challenge sessions, replace the commitment proofs produced by honest endpoints with simulated proofs, while leaving the ciphertext-statement proofs produced by the hubs unchanged. By the zero-knowledge property of $\Pi_{\mathrm{ZK}}$, $\mathcal{H}_0$ and $\mathcal{H}_1$ are computationally indistinguishable.

\emph{Hybrid $\mathcal{H}_2$.} In $\mathcal{H}_1$, change the association between credential issuance and presentation from $\mathcal{M}_0$ to $\mathcal{M}_1$. The blindness and randomizability of $\Pi_{\mathrm{BRS}}$ make this change indistinguishable to the participating hubs.

\emph{Hybrid $\mathcal{H}_3$.} Change the matching of the randomized condition pairs from $\mathcal{M}_0$ to $\mathcal{M}_1$. Let $x_{0,j}$ be the witness generated for $R_j$ and $z_i$ the final witness on the sending side of $S_i$. Under either matching, independently sample $z_i,\alpha_{R,j}\xleftarrow{\$}\mathbb{Z}_q^*$ and define $\alpha_{S,i}=z_i(\alpha_{R,j}x_{0,j})^{-1}$. Since $x_{0,j}\in\mathbb{Z}_q^*$, the map $(\alpha_{R,j},\alpha_{S,i})\mapsto(\alpha_{R,j},z_i=\alpha_{S,i}\alpha_{R,j}x_{0,j})$ is a bijection on $(\mathbb{Z}_q^*)^2$. The sampled values therefore have the same joint distribution as the two independent randomizers in the real protocol, and $x_{2,i}=z_i$. In the LOE model, the resulting sender-side pair is distributed as $(z_iG,\mathsf{Enc}(pk_{\mathrm{LOE}},z_i))$ with fresh encryption randomness and independently of the receiver matched to $S_i$. Throughout this change, the hubs' keys, receiver-side witnesses, and initial ciphertexts remain fixed. The intermediate pair and witness pass through the anonymous endpoint channel. During $\mathsf{Solve}$, sender-side tasks use $z_i$ and receiver-side tasks use $x_{0,j}$, in agreement with the sampled randomizers.

Couple the two matchings using the same local signing and channel realization randomness. Their local conditions, public inputs, and phase schedules coincide, so each channel derives candidate states, task tuples, and evidence retention and refresh sets with the same distribution. The resulting task evidence, state confirmation signatures, claim outputs, and public ledger observations also have the same distribution. Thus, $\mathcal{H}_2$ and $\mathcal{H}_3$ have the same distribution in the LOE model.

\emph{Hybrid $\mathcal{H}_4$.} Replacing the simulated endpoint proofs with real proofs yields the real transcript generated with $\mathcal{M}_1$. By the zero-knowledge property of $\Pi_{\mathrm{ZK}}$, $\mathcal{H}_3$ and $\mathcal{H}_4$ are computationally indistinguishable.

Credential reassociation and the simulation or restoration of endpoint proofs introduce only negligible distinguishing advantages under the blindness and randomizability of $\Pi_{\mathrm{BRS}}$ and the zero-knowledge property of $\Pi_{\mathrm{ZK}}$, respectively. Reassigning the condition pairs preserves their distribution in the LOE model. Hence, the real transcripts for $\mathcal{M}_0$ and $\mathcal{M}_1$ are computationally indistinguishable, and the advantage in \eqref{eq:srbcs-blindness} is negligible.
\end{proof}

\begin{proof}[Proof of Theorem~\ref{thm:srbcs-construction}]
The four preceding lemmas establish the correctness, task unforgeability, unlockability, and blindness requirements of Definition~\ref{def:srbcs-security}. Therefore, $\Pi_{\mathrm{S\text{-}RBCS}}^{*}$ is secure.
\end{proof}

\section{Additional Experiments}
\label{sec:additional-experiments}

\subsection{Evaluation on Bitcoin and Litecoin}
\label{sec:experiment-bitcoin-litecoin}

We evaluate UniPCH on three Bitcoin and three Litecoin instances running in regtest mode. Each blockchain has 50 users, and each user initiates 5 payments. The channels between hubs use Bitcoin. 

Table~\ref{tab:bitcoin-litecoin-comparison} reports the average running time and bandwidth per payment in each protocol phase.

\begin{table}[ht]
    \centering
    \caption{Running time and bandwidth of UniPCH on Bitcoin and Litecoin.}
    \label{tab:bitcoin-litecoin-comparison}
    \footnotesize
    \setlength{\tabcolsep}{1.0pt}
    \renewcommand{\arraystretch}{1.05}
    \begin{tabularx}{\columnwidth}{|>{\hsize=1.2\hsize\linewidth=\hsize\centering\arraybackslash}X|>{\hsize=.8\hsize\linewidth=\hsize\centering\arraybackslash}X|>{\centering\arraybackslash}X|>{\centering\arraybackslash}X|>{\centering\arraybackslash}X|>{\centering\arraybackslash}X|>{\centering\arraybackslash}X|}
    \hline
    \multicolumn{2}{|c|}{\textbf{Running Time (ms)}} & $\mathsf{Reserve}$ & $\mathsf{Promise}$ & $\mathsf{Convert}$ & $\mathsf{Solve}$ & Total \\
    \hline
    \multirow{2}{*}{UniPCH} & BN & 3.374 & 446.9 & 136.2 & 156.2 & 742.6 \\
    \cline{2-7}
     & BF & 5.381 & 489.7 & 251.9 & 309.8 & 1056.8 \\
    \hline
    \end{tabularx}
    \vspace{0.5em}

    \begin{tabularx}{\columnwidth}{|>{\hsize=1.2\hsize\linewidth=\hsize\centering\arraybackslash}X|>{\hsize=.8\hsize\linewidth=\hsize\centering\arraybackslash}X|>{\centering\arraybackslash}X|>{\centering\arraybackslash}X|>{\centering\arraybackslash}X|>{\centering\arraybackslash}X|>{\centering\arraybackslash}X|}
    \hline
    \multicolumn{2}{|c|}{\textbf{Bandwidth (KB)}} & $\mathsf{Reserve}$ & $\mathsf{Promise}$ & $\mathsf{Convert}$ & $\mathsf{Solve}$ & Total \\
    \hline
    \multirow{2}{*}{UniPCH} & BN & 1.435 & 14.79 & 8.652 & 4.214 & 29.09 \\
    \cline{2-7}
     & BF & 2.539 & 32.28 & 43.58 & 43.27 & 121.7 \\
    \hline
    \end{tabularx}
\end{table}

The running times are close to those of XPCH under the same deployment. UniPCH uses 18.4\% and 30.0\% more bandwidth on the Near and Far paths, respectively, to carry task evidence and additional protocol fields. The Far path incurs higher costs in $\mathsf{Convert}$ and $\mathsf{Solve}$ because updates to the UTXO channels between hubs refresh evidence for active transfer tasks.

\subsection{On-chain Costs}
\label{sec:onchain-cost}

We evaluate the on-chain costs of UniPC's UTXO and account realizations. We first measure the costs of channel operations and then examine how the number of active items affects dispute settlement. For account channels, we compare ordered lists and Merkle roots and measure Gas consumption using transaction receipts. For Bitcoin channels, we measure the size of fully signed transactions in bytes and virtual bytes.

\noindent\textbf{Channel Operation Costs.}
Table~\ref{tab:onchain-cost} shows the costs of deployment, channel opening, cooperative closing, and dispute operations. For account channels, opening includes channel creation and both deposits, while closing includes cooperative settlement and both withdrawals. For Bitcoin, opening and closing correspond to the funding and cooperative closing transactions. For dispute settlement, we measure the submission of a state containing one active item and the corresponding claim or refund. Bitcoin state submission corresponds to the commitment transaction. The two account representations have similar opening costs, while the ordered list uses less Gas for closing.

\begin{table}[ht]
    \centering
    \caption{On-chain costs of UniPC operations.}
    \label{tab:onchain-cost}
    \footnotesize
    \setlength{\tabcolsep}{2.0pt}
    \renewcommand{\arraystretch}{1.08}
    \renewcommand{\tabularxcolumn}[1]{m{#1}}
    \begin{tabularx}{\columnwidth}{|>{\hsize=1.25\hsize\linewidth=\hsize\centering\arraybackslash}X|*{4}{>{\hsize=.9375\hsize\linewidth=\hsize\centering\arraybackslash}X|}}
    \hline
    \multirow{2}{*}{Operation} & \multicolumn{2}{c|}{Ethereum} & \multicolumn{2}{c|}{Bitcoin} \\
    \cline{2-5}
     & List Gas & Merkle Gas & Tx Bytes & Tx vBytes \\
    \hline
    $\mathsf{Deploy}$ & 2,104,995 & 2,013,101 & --- & --- \\
    \hline
    $\mathsf{Open}$ & 255,753 & 255,368 & 351.0 & 189.0 \\
    \hline
    $\mathsf{Close}$ & 183,595 & 202,690 & 334.3 & 169.0 \\
    \hline
    $\mathsf{Submit}$ & 182,160 & 158,163 & 377.0 & 211.7 \\
    \hline
    $\mathsf{ClaimTransfer}$ & 71,870 & 98,354 & 352.3 & 152.0 \\
    \hline
    $\mathsf{RefundTransfer}$ & 63,761 & 90,506 & 275.3 & 131.0 \\
    \hline
    $\mathsf{RefundEscrow}$ & 63,818 & 90,611 & 200.7 & 112.0 \\
    \hline
    \end{tabularx}
\end{table}

\begin{figure}[ht]
    \centering
    \includegraphics[width=\columnwidth]{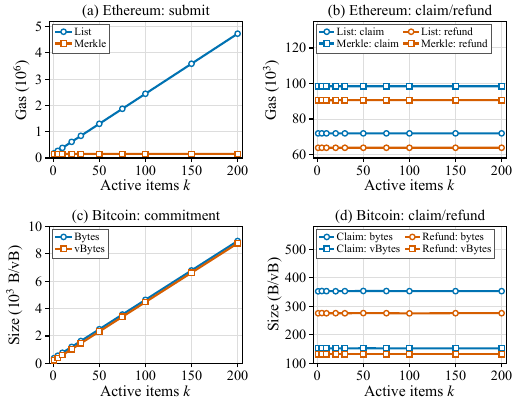}
    \caption{On-chain costs for different numbers of active items.}
    \label{fig:onchain-cost}
\end{figure}

\noindent\textbf{Effect of Active Items.}
We vary the number of active items $k$ over $\{1,5,10,20,30,50,75,100,150,200\}$ and measure the costs of state submission, $\mathsf{ClaimTransfer}$, and $\mathsf{RefundTransfer}$. Fig.~\ref{fig:onchain-cost}(a) and~(b) show the Gas costs for account channels. The ordered list is submitted and stored in full. Its state submission cost rises from 182,160 to 4,727,766 Gas as $k$ increases from 1 to 200. The Merkle representation stores the item commitment as a single root, keeping the submission cost around 158,000 Gas.

The costs of $\mathsf{ClaimTransfer}$ and $\mathsf{RefundTransfer}$ change little as $k$ increases. Ordered lists require about 71,900 and 63,800 Gas for these operations, respectively. The corresponding costs with Merkle roots are about 98,400 and 90,600 Gas. Merkle claims and refunds require membership proofs, while ordered lists use item indices. Our user channels use ordered lists for simple maintenance and lower claim and refund costs. Channels between hubs use Merkle roots to keep state submissions compact as active items accumulate.

Fig.~\ref{fig:onchain-cost}(c) and~(d) show the transaction sizes for Bitcoin. The commitment transaction grows from 377.0 to 8,933.7 bytes as $k$ increases from 1 to 200. Its virtual size increases from 211.7 to 8,768.7 vBytes. Each additional item introduces an output of 43 bytes. Transactions for $\mathsf{ClaimTransfer}$ and $\mathsf{RefundTransfer}$ each spend a single item output, so their sizes remain stable. They are about 353 and 276 bytes, respectively, with virtual sizes of 152 and 131 vBytes. $\mathsf{RefundEscrow}$ uses a shorter refund script and requires about 201 bytes with a virtual size of 112 vBytes.

\end{document}